\documentclass[format=acmsmall, review=false]{acmart}
\usepackage{acm-ec-23}
\usepackage{booktabs} 
\usepackage[ruled]{algorithm2e} 

\SetAlFnt{\small}
\SetAlCapFnt{\small}
\SetAlCapNameFnt{\small}
\SetAlCapHSkip{0pt}
\IncMargin{-\parindent}

\setcitestyle{acmnumeric}

\title[]{Strategy-proof Budgeting via a VCG-like Mechanism}

\author{Jonathan Wagner, Reshef Meir}

\begin{abstract}
We present a strategy-proof public goods budgeting mechanism where agents determine both the total volume of expanses and the specific allocation. It is constructed as a modification of VCG to a non-typical environment, namely where we do not assume quasi-linear utilities nor direct revelation. We further show that under plausible assumptions it satisfies strategyproofness in strictly dominant strategies, and consequently implements the social optimum as a Coalition-Proof Nash Equilibrium. A primary (albeit not an exclusive) motivation of our model is Participatory Budgeting, where members of a community collectively decide the spending policy of public tax dollars. While incentives alignment in our mechanism, as in classic VCG, is achieved via individual payments we charge from agents, in a PB context that seems unreasonable. Our second main result thus provides that, under further specifications relevant in that context, these payments will vanish in large populations. In the last section we expand the mechanism’s definition to a class of mechanisms in which the designer can prioritize certain outcomes she sees as desirable. In particular we give the example of favoring equitable/egalitarian allocations.                      
\end{abstract}
\usepackage{graphicx}
\usepackage{physics}
\usepackage{refcount}
\usepackage{amsfonts}
\usepackage{url}
\usepackage{mathrsfs}
\usepackage{graphics}
\usepackage{mathtools}
\usepackage{thmtools,thm-restate}
\usepackage{enumitem}
\usepackage{verbatim}
\usepackage{relsize}
\usepackage{tikz}
\usepackage{hyperref}

\declaretheorem[style=definition]{assumption}
\declaretheorem[style=definition]{observation}

\DeclareMathOperator*{\argmax}{arg\,max}
\DeclareMathOperator*{\argmin}{arg\,min}
\usepackage{bbm}

\def\R{\mathbb{R}}
\def\N{\mathbb{N}}
\def\M{\mathcal{M}}
\def\I{\mathcal{I}}
\def\V{V}
\def\D{\mathcal{D}}

\def\hx{\hat{x}}
\def\hM{\hat{\mathcal{M}}}

\def\ha{\hat{\alpha}}
\def\ba{\bar{\alpha}}
\def\hT{\hat{\Theta}}
\def\ff{\mathcal{C}}

\def\hg{\hat{g}}
\def\hP{\hat{P}}
\def\he{\hat{e}}

\def\th{\hat{t}}
\def\ds{\Delta^m \times [-\frac{B_0}{n},\infty)}
\def\tr{[-\frac{B_0}{n},\infty)}
\def\td{\Delta^m \times \R_{++}}

\newenvironment{rtheorem}[1]{\medskip\noindent\textbf{Theorem~\ref{#1}. }\begin{itshape}}{\end{itshape}\medskip}

\newenvironment{rlemma}[1]{\medskip\noindent\textbf{Lemma~\ref{#1}. }\begin{itshape}}{\end{itshape}\medskip}
\newenvironment{rcorollary}[1]{\noindent\textbf{Corollary~\ref{#1}. }\begin{itshape}}{\end{itshape}}

\newenvironment{rexamplecom}[2]{\smallskip\textit{Example~\ref{#1} (#2). }}{\medskip}

\def\I{\mathcal{\large I}}
\def\1q2{\frac{1}{2}}

\newcommand{\newpar}[1]{\vspace{-0mm}\paragraph{\textbf{#1.}}}
\begin{document}

\begin{titlepage}

\maketitle

\end{titlepage}

\section{Introduction}
We study a model where a population of $n$ agents face the decision of funding several pubic goods or investments that serve their collective objectives. Formally, we start with an available budget of $B_0 \geq 0$ that should be allocated among $m\geq 1$ different alternatives. The \emph{budget decision} we seek is a pair $(x,t)$ where  $t \in \tr$ is a monetary sum ("tax") that each agent adds to (or subtracts from) the budget, and  $x \in \Delta^m:=\{x \in R^n | x_j \geq 0\ \forall j,\ \sum_j x_j =1\}$ represents the allocation of the resulting budget $B_t = B_0 +nt$ among the $m$ alternatives. Note that we allow $t<0$, meaning that some of the initial budget $B_0$ will be distributed equally among agents rather than fund public investments. We are interested in constructing a collective decision mechanism to which every agent submits her preferred budget decision, and one that is incentive compatible (IC) particularly. Table~\ref{table:profile_ex} demonstrates how agents might report  to such mechanism their preferences for budget allocation among three municipal services.
Note e.g. that agent $a$ suggests an individual payment of $t=20$ each, and so
to allocate a total of $B_0+3t=100+3\cdot20=160$. Also note that agents $b$
and $c$ propose the same normalized allocation $x_b=x_c \in \Delta^3$ but
differ in their preferred tax.

\begin{table}[t]
    \centering
   \begin{tabular}{l|c|c|c|c}
& Tax ($t$) & Education $(x_1)$ & Parks $(x_2)$ & Transport $(x_3)$ \\
\hline
agent~$a$   & 0 & 60 (0.6) & 30 (0.3) & 10 (0.1) \\
agent~$b$  & 20  & 32 (0.2)& 48 (0.3) & 80 (0.5)\\
agent~$c$ & -20 & 8 (0.2) & 12 (0.3) & 20 (0.5) \\
\end{tabular}\vspace{.2cm}

    \caption{An example voting profile for three alternatives and three agents, with an initial budget $B_0=100$.\vspace{-4mm}}
    \label{table:profile_ex}
\end{table}

\subsection{Motivation and Main Results}
Incentives alignment has long been a major interest in Social Choice and Mechanism Design research. While Gibbard \& Satterthwaite's Theorem \cite{gibbard1973manipulation,satterthwaite1975strategy} provides a negative result for the most general setup, positive results do exist when the preference space is contracted somehow (mainly to singled-peaked~\cite{nehring2007structure}), or when monetary transfers that the mechanism charges from agents are introduced. In the latter case, the VCG mechanism \cite{clarke1971multipart} is the canonical model and moreover, Roberts' Theorem \cite{roberts1979characterization}  shows that for a general preferences domain, any strategy-proof mechanism is in a sense a generalized version of the VCG principle. However, VCG is built upon an assumed model that cannot apply to every scenario. Most importantly, it depends crucially on underlying quasi-linear utilities, meaning that the overall satisfaction of agent $i$ when the mechanism outputs the decision $\Omega$ and charges her with a payment $p_i$ is expressed as 
$u_i(\Omega,p_i)=v_i(\Omega)-p_i$ where $v_i(\Omega)$ is the value she attributes to $\Omega$. While plausible in many economic situations, quasi-linearity is in particular violated if agents do not  asses their benefits from different outcomes in monetary terms (making the subtraction of $p_i$ from $v_i(\Omega)$ a meaningless expression), whether because these are purely non-financial or just difficult to quantify. When we consider investments in public goods, that is likely to be the case. 
Furthermore, VCG is a \emph{'Direct Revelation'} mechanism in which agents explicitly report the full description of their preferences~\cite{nisan2007algorithmic}. That is reasonable when, for example, we ask agents to specify a value (that is, the maximum price they are willing to pay) for each of the items offered in an auction. In our case the space of optional outcomes is a multidimensional continuum and the 'value' is quite an abstract concept, non-linear in particular. As should be obvious when we later introduce our model to its full details, a direct revelation mechanism becomes unlikely in that context and we thus only collect optimal allocations,\footnote{Possibly proceeded by a few follow-up questions presented to agents, See \ref{sec_mech}.} as demonstrated in Table~\ref{table:profile_ex}.  

 \newpar{Contribution and paper structure} We  show that for the budgeting problem described above we can construct a 'VCG-like' mechanism all the same. The key idea we will use is that while we can no longer justify the simple quasi-linear relation between the utility in the outcome and that of the charged payment, the fact that the decision space itself involves monetary transfers enables revealing of the true relation between these two in each agents' perspective. At the begining of Section \ref{sec_mech} show how full information can in fact be extracted from the preferences that agents report. Section \ref{sec_USVCG} introdoces our proposed mechanism, and in \ref{sec_IC} we show that under certain conditions it is IC in strictly dominant strategies and consequently Coalition-proof \cite{bernheim1987coalition}. Finally, in some applications of the model collecting money from agents (that is, money paid "to the mechanism", on top of the tax $t$ that is a part of the chosen outcome and finances shared expenditures) serves well the purpose of aligning incentives, however may not be reasonable on its own merits. Thus, our second and most technically challenging main result, presented in Section \ref{sec_pconv}, provides that, for a relevant subclass of utility functions, these payments become negligible in large populations . 
 
A few examples of environments where our mechanism may be useful are listed below, starting with the one standing at the center of our focus.

\newpar{Participatory Budgeting} Our model falls naturally within a Participatory Budgeting (PB) framework, where members of a community determine the allocation of their shared budget via some voting mechanism.\footnote{’The Participatory Budget Project’: https://www.participatorybudgeting.org/} While most PB implementations worldwide, and accordingly much of the PB literature, involve the allocation of a given budget among several public indivisible goods, several previous works were dedicated to divisible PB models  \cite{freeman2019truthful,goel2019knapsack,brandl2020funding,fain2016core,garg2019iterative}, where `projects' may correspond to  city departments, e.g. education, transportation, parks and recreation and so on. Our model moreover involves taxation in the collective decision, in which we find several advantages. Technically, our adjustment of VCG payments to non quasi-linear utilities is enabled solely due to that feature. Conceptually, it expands the scope of control delegated to society through direct democracy, and also brings a more valuable feedback to the designer on the true value that public expenditures create, rather than just a comparison between them.  Being our primary motivation, the terminology we use and our approach to modelling is mainly driven by the PB application, however intended to represent a wider range of real life economic situations. While we are aware that some of the assumptions we make may not be generically adequate in every possible instance of different environments too, the basic ideas extend to at least some of them. 
We thoroughly discuss our works' relations with previous models and results in that area in the next section of the introduction.

\newpar{Shared financing among nearby businesses} A group of nearby businesses (for example in the same shopping center) might cofinance some facilities and services that are more public or regional by nature, e.g. security, costumers' parking, shared dining areas, public activities that increase exposure, etc.  

\newpar{Environmental or other non-financial investments} That could apply to governments imposing an 'environmental (or other) tax' on firms or different countries deciding on their coordinated actions \cite{bjorvatn2002tax}. Other pursued goals might include border controls of neighbouring countries, for example. 

\newpar{Joint R\&D and Human Capital ventures} R \& D and Human Capital investments are typically long term and require considerable resources. Thus, firms in the same or related industries might benefit from joining forces in funding e.g. the training of required professions or the development of new technologies and methods. Such collaborations might scale from a small number of businesses deciding to cooperate, through a unionized industry and up to being run by the government state-wide.

\newpar{Non-monetary 'Currencies'}
As the concept of 'value' of different outcomes is more abstract in our model, it may also apply to situations where the investments themselves, and thereby the collected tax payments, are not necessarily monetary. Examples may include members in a community willing to dedicate an agreed amount of working hours within their community, or firms that allocate resources such as land (e.g. agricultural land dedicated to cooperative experimental planting), or technological resources such as storage space or computation force.\\

A variation of our model where we employ a heterogeneous taxation may also be relevant, especially to applications other than PB. For example, imposing higher contributions on wealthier countries in joint environmental investments. We discuss that in appendix \ref{appsec_heterotax} and show that our IC results extend to such setup. 
\subsection{Modeling Approaches in Divisible Participatory Budgeting}
Before presenting our own model, we discuss here the limitations of existing incentive compatibility concepts from the divisible PB literature. 

\newpar{Spatial models}
Several past works have studied incentive alignment in PB and divisible PB particularly {\cite{aziz2019fair,freeman2019truthful,goel2019knapsack,fain2016core,bhaskar2018truthful}}. Notably, two different works~\cite{freeman2019truthful,goel2019knapsack} presented incentive compatible (or 'strategyproof') mechanisms for a divisible PB scenario similar to ours. These works assume $\ell_1$-norm preferences in which each agent $i$ reports her optimal allocation $x^{(i)} \in \Delta^m$ to the mechanism, and her utility in any allocation $x$ is given by $u_i(x)=-\norm{x^{(i)}-x}_1$. Other $\ell_p$ norms~\cite{garg2019iterative}, or more generally single-peaked preferences \cite{moulin1980strategy,barbera1994characterization}, are well studied in social choice literature and known to allow for mechanisms with strong strategyproofness guarantees. This IC definition relies on the underlying assumption that the utility of an agent depends solely on the `distance' between the accepted outcome and her favoured alternative. Indeed, in the absence of a concrete way to measure the monetary value of decisions, minimizing the distance to a voter's ideal allocation is a reasonable solution.

However, we argue that when agents have concrete measurable utility from decisions, as in PB, the spatial approach may not adequately capture the preference of a voter, as the very reason that voters' optimal allocations are not typically uniform is that they value money spent on each alternative differently. In particular, minimizing the $\ell_1$ distance is not just suboptimal, but may in fact incentivize agents to manipulate .
To see why, consider  Voter~$a$ from Table~\ref{table:profile_ex}, with her ideal allocation $x_a = (0.6,0.3,0.1)$. Now, let us further assume that  $a$ is the parent to children under 18 years of age and the  three public goods considered are education systems, public parks and public transportation--- that she very rarely uses and thus does not value highly. Reasonably, that voter might strictly prefer the allocation $x'=(0.7,0.2,0.1)$ over $x''=(0.6,0.2,0.2)$ because while both are suboptimal to her, investing another $0.1$ of the budget in the education of her children may serve her interests better than investing it in a facility she rarely enjoys. 
However, under the $\ell_1$-distance model she is indifferent between the two, meaning that her incentive to manipulate the outcome from $x''$ to $x'$, when and if facing the opportunity, is not accounted for.\footnote{Indeed, concrete examples that show how the incentive compatibility in \cite{goel2019knapsack} and \cite{freeman2019truthful} might ``fail", in the above sense, are not difficult to construct. }
   
\newpar{Social choice vs. microeconomic models}
As explained in \cite{freeman2019truthful}, the solution concept that minimizes the $\ell_1$ distances from agents' ideal points  generalizes the one-dimensional median rule \cite{moulin1980strategy}. Similarly, most of the literature on divisible PB adopt or generalizes the same assumptions used for PB with indivisible projects, which in turn stem from standard modelling assumptions in voting and computational social choice. 

However, we argue that divisible budget allocation is much closer in nature to problems typically treated within microeconomic framework \cite{stiglitz1977theory,wildasin1988nash,bernard1977optimal}. This is true especially when assigning budgets to departments etc. rather than to specific projects with variable costs.\footnote{See e.g. \url{https://pbstanford.org/boston16internal/knapsack}}  Hence, it makes more sense to adopt conventional economic assumptions regarding demand and utility, as in \cite{fain2016core} and as we do here. In particular:
\begin{itemize}
\item \emph{Additive concave utilities.} We adopt the additive concave utilities model \cite{fain2016core,jain2007eisenberg,vazirani2011market} that offers a more concrete description of the utility gained from different public investments. Its most closely related version to ours is found in a former work by Fain et al.~\cite{fain2016core}. There, the utility of agent $i$ in allocation $X=(X_1,X_2\dots)$ is expressed as 
\begin{equation}\label{eq_additive_model}
    U_i(X)=\sum_j \alpha_{i,j}\theta_j(X_j)
\end{equation}    where $X_j$ is the amount spent on public good $j$, the $\{\theta_j\}_j$ functions are monotonically increasing in $X_j$ for all $j$ and strictly concave, (smoothly) expressing the plausible assumption of decreasing marginal gains, and $\alpha_{i,j}$ are scalars that vary between agents. As we assume that part of the budget is collected via a tax-per-agent $t$, our model adds on the above the disutility  of a voter from the tax payment.
\item \emph{Optimal points characterized by the MRS conditions}, that follows form the concavity and some additional conventions on utilities \cite{krugman2008microeconomics}.
\item \emph{Utility depends on public investment per capita.} (that we add to the model in Section \ref{sec_pconv}) . On a large scale, it is reasonable that the quality of public goods depends more on spending per capita rather than on the nominal amount.\footnote{See for example \url{https://data.oecd.org/gga/general-government-spending.htm}, and \cite{stiglitz1977theory}.}
\end{itemize}
In contrast, \emph{elicitation} is an issue that has received much more attention in the literature of mechanism design and computational social choice than in microeconomics. For example there is a live discussion in the context of indivisible PB on the tradeoff between expressiveness of the ballot and effort required for elicitation~\cite{fairstein2019modeling,benade2021preference}.  Similarly, we argue that it does not make sense to assume that we have direct access to voters' preferences, and here we adopt from computational social choice the assumptions that voters simply report their most preferred allocation, as in~\cite{freeman2019truthful}. 

\medskip

In terms of applicability, however, the obvious shortcoming of our model is that it requires us to explicitly specify the functions $\{\theta_j\}_j$ and $f$, which are fairly abstract. Importantly, we \textbf{do not} assume that agents are 'aware' of their assumed utility function, but, conventionally, only know their preferences regarding the decision space, that presumably can be interpreted as derived from an underlying utility model \cite{krugman2008microeconomics}. Of course, any such model would be an approximation at best.  Nevertheless, it is fair to assume that any choice of monotonically increasing concave functions probably better approximates individuals' preferences---and thereby incentives---than the spatial model or a linear additive model \cite{brandl2020funding}. (Note that the linear additive model is arguably much less reasonable than concave, not merely due to  excluding diminishing returns, but because it implies boundry optimal allocations where only alternatives $j$ that maximize $\alpha_{i,j}$ are implemented).

\subsection{Further related literature}\label{introsec_RL}
The Economic literature on public goods markets, equilibria and optimal taxation is abundant. (\cite{stiglitz1977theory,vazirani2011market,bernard1977optimal}, just to name a few). While our work adopts a similar approach to modelling and also optimizes social welfare, this brunch of the literature rarely discusses mechanisms. One exception that we know of is found in \cite{lahkar2020dominant}, in which the socially optimal outcome is implemented in strictly dominant strategies using a method very similar to ours, however for quite a different utility model.     
To the best of our knowledge, the only existing PB mechanism that included tax in the collective decision previously to ours was studied by  Garg et al. \cite{garg2019iterative} in the context of experimenting \emph{'iterative voting'} mechanisms. Interestingly, it may suggest some supporting evidence in favour of the additive concave utility model over spatial models in that context. Two other previous works \cite{aziz2021participatory,brandl2020funding} incorporated private funding into a PB model, albeit in the form of voluntary donations that every agent can choose freely and not as a collectively decided sum that is collected from (or paid to) everyone, as we consider here. 


The literature on divisible PB is relatively narrow. In terms of incentive compatibility,  \cite{freeman2019truthful} presented the soundest results, under a spatial utility model. 
Alternatively, Fain et al.~\cite{fain2016core} offer a randomized mechanism that is `approximately-truthful' for the special case of 1-degree homogeneous additive utilities. The \emph{Knapsack Voting} mechanism introduced in \cite{goel2019knapsack} also satisfies some weaker notion of strategyproofness under a similar model. Aziz et al. \cite{aziz2019fair} presented IC mechanisms for additive linear utilities, although their model is primarily motivated by randomized approval mechanisms. A similar utility model is also found in \cite{talmon2019framework}.
Overall, in relation to the divisible PB field, this work offers an SDSIC mechanism under concave additive utilities, to the best of our knowledge for the first time. 

Our desire for diminishing the (modified) VCG payments resembles the idea of \emph{redistribution} in mechanism design \cite{cavallo2006optimal,guo2009worst}. Such methods are especially relevant in a discrete decision space and can eliminate surplus only partially, while in our model the complete (asymptotic) vanishing is much thanks to the continuity of the decision space.

Much of the PB literature deals with the concept of \emph{fair allocations}~\cite{fain2016core,peters2020proportional,aziz2017proportionally}. While not a primary goal of our model, we show that the designer can bias the allocation closer to a favorable allocation---including one they see as fair.
\medskip

\section{Model and Preliminaries}\label{sec:prelim} We denote by $\Delta^m$ the set of distributions over $m$ elements, and use $[m]$ as a shortcut for $\{1,\ldots,m\}$. 
A set of $n$ agents (voters) need to collectively reach a \emph{budget decision} $(x,t)$ described as follows. $t\in \R$ is a lump-sum tax collected from every agent. $x \in \Delta^m$ is an allocation of the total available budget $B_t:=B_0+nt$ among some $m$ pre-given public goods, where $B_0$ is an external, non tax-funded source of public funds. $t$ is restricted only by $t>-\frac{B_0}{n}$, meaning that voters can decide either to finance a budget larger than $B_0$ through positive taxation, or allocate some of it to themselves directly as cash (negative taxation). The collective decision is taken through some voting mechanism to which every agent submits her most preferred budget decision $(x^{(i)},t^{(i)})$ .

\subsection{Preferences}\label{sec_pref}
We now introduce the utility function step by step. We start with agents' valuation for public expenditures alone that follows from \cite{fain2016core}:
$$ \Theta_i(X) :=  \sum_{j=1}^m \alpha_{i,j}\theta_j(X_j)$$
where $X=(X_1,\dots,X_m) \in \R_+^m$ and $X_j$ is the amount spent on project $j$. For all $j \in [m]$, an agent $i$ gains $\alpha_{i,j}\theta_j(X_j)$ from an $X_j$ spending on public good $j$. $\{\theta_j\}_{j=1}^m$ are identical across agents while agents differ one from another in their coefficients $(\alpha_{i,1}\dots,\alpha_{i,m}) \in \Delta^m$.
\begin{assumption}\label{theta_asm}
For all $1 \leq j \leq m$, $\theta_j:\D \to \R, \D \in \{\R_+,\R_{++}\}$  is increasing and strictly concave, and $\lim_{X_j \to 0}\theta(X_j) \leq 0$, where $\R_+ := \{x \in \R | x \geq 0\},\ \R_{++} := \{x \in \R | x > 0\}$.
\end{assumption}

As explained earlier in the Introduction, our model includes monetary transfers of two types. One is the collectively decided tax that will be collected from every agent, and beyond that, a mechanism may charge payments in order to align incentives, where these payments are not affecting public expenditure. Thus we include in our model the value agents attribute to reduction (or increment) in their overall wealth,
$$\pi_i(\delta_w) := -\alpha_{i,f}\cdot f(\delta_w)$$
where $\delta_w$ is monetary loss (gain when negative) and the coefficient $\alpha_{i,f}>0$ varies between agents. 
Our formulation of the disutility function $f$ will follow Kahneman and Tversky's Prospect Theory \cite{tversky1992advances}. Their and others' empirical findings ~\cite{tversky1992advances,booij2010parametric} demonstrate that people tend to exhibit \emph{loss aversion} in relation to monetary gains or losses, meaning: (a) valuating monetary transfers with reference to their current status presumably located at the origin; and (b) exhibit risk-aversion with respect to monetary losses, and risk-seeking with respect to gains. Meaning, we shall assume that $f(0)=0$ and that it is strictly convex in $(-\infty,0]$ and strictly concave in $[0,\infty)$. In principle, our analysis requires differentiable utility functions. However, the most natural examples of elementary increasing concave functions, the logarithmic and power functions, are not differentiable at zero, which still allows for our results. Thus, for the sake of giving more intuitive and simple examples, we will allow a diverging derivative at zero.\footnote{Moreover, Kahneman and Tversky themselves suggested the following explicit functional form  (the\emph{"value function"} in their terminology) built on power functions, that has been adopted widely ever since \cite{booij2010parametric}:
\begin{equation*}\label{vf_KT}
    f^{KT}_{q,r,\lambda}(\delta_w)=-\mathbbm{1}_{\{\delta_w \leq 0\}}|\delta_w|^q+\lambda \cdot \mathbbm{1}_{\{\delta_w >0\}}(\delta_w)^r
\end{equation*}
for some $0 <q,r<1$, $\lambda>0$ (See \cite{tversky1992advances}, p. 309). Assumption \ref{f_asm} thus includes the above form and extends beyond it.}

\begin{assumption}\label{f_asm}
$f:\R \to \R$ is increasing, strictly convex in $(-\infty,0]$ and strictly concave in $[0,\infty)$. $f(0)=0$. We assume that $f$ is either continuously differentiable in all $\R$ or anywhere but the origin, in which case $\lim_{z \to 0^+_-}f'(z)=\infty$.
\end{assumption}
Now, by adding $\Theta_i$ and $\pi_i$ together we get the full description of an agent's
\emph{utility function}: 
$$ u_i(X,\delta_w) :=  \Theta_i(X)+\pi_i(\delta_w)= \sum_{j=1}^m \alpha_{i,j}\theta_j(X_j)-\alpha_{i,f}f(\delta_w). $$
In particular, the problem at hand is reaching a collective \emph{budget decision} $(x,t)$ via some voting mechanism that aggregates the collective preferences on the whole decision space $\ds$. We therefore specifically define an agent's valuation for a budget decision (i.e. with no regard to payments on top of the collected tax):
\begin{align*}
   \forall x \in \Delta^m,\ t \in \R,\quad 
   v_i(x,t):&= u_i(x \cdot B_t,t) = \sum_{j=1}^m \alpha_{i,j}\theta_j\big(x_j \cdot B_t\big)-\alpha_{i,f}f(t)
\end{align*}
Thus the \emph{type} of $i$ is defined by the coefficients:
\begin{equation}\label{def_pvector}
    \alpha_i:=(\alpha_{i,1}\dots,\alpha_{i,m},\alpha_{i,f}) \in \Delta^m \times \R_{++}
\end{equation}
In a more general sense, we also write $u_\alpha$, $\Theta_\alpha$ and $v_\alpha$ for functions of a hypothetical "type $\alpha$" agent.

Finally, we reasonably want to assume that every agent would like to fund \emph{some} level of public expenditures, whereas no agent favors an infinite budget funded by infinite tax.
\begin{assumption}\label{theta_f_asm} 
\phantom{gg}
\begin{itemize}
    \item $\forall i\, \exists \, j$ s.t. $\lim_{t \to 0}n\alpha_{i,j}\theta'(nt) > \alpha_{i,f}\lim_{t \to -\frac{B_0}{n}}f'(t)$.
    \item For any $m \in \N$, $\lim_{z \to \infty}\frac{\theta'_j(z/m)}{f'(z)}=0\ \forall j \in [m]$.
\end{itemize}
\end{assumption} 
Hence, a  \emph{budgeting instance} $\I=\{m,n,B_0, \vec{\alpha},\{\theta_j\}_{j\in[m]},f\}$ is defined by the number of public goods $n$, number of agents $m$, initial budget $B_0$, the \emph{type profile} $\vec{\alpha}=(\alpha_i,\dots,\alpha_n)$ and functions $\{\theta_j\}_{j\in[m]},f$ that respect Assumptions \ref{theta_asm},\ref{f_asm},\ref{theta_f_asm}. For example,
\begin{example}[Running example] \label{ex:run}
 Consider an instance with $B_0=0,m=2,n=3$, where for every agent $i \in [3]$, the valuation is:
 \vspace{-2mm}
\begin{equation}\label{eq:v_runex}
  v_i(x,t):=10\cdot \sum_{j=1}^2\alpha_{i,j}\ln(x_j\cdot 3t)-\alpha_{i,f}\sqrt{t}  
\end{equation}
That is,  $\quad \theta_j(X_j):=10\ln(X_j)$ for both $j\in\{1,2\}$, and $f(t):=\sqrt{t}$.
 \end{example}

\section{mechanisms}\label{sec_mech}
\medskip

\newpar{Preference Elicitation}
The mechanism we introduce in the next section is designed to maximize the social welfare $\sum_iv_i(x,t)$. However, agents do not report their full valuations $v_i(\cdot,\cdot)$ explicitly but rather their optimal budget decisions $(x^{(i)},t^{(i)})\ \forall i \in [n]$. Eliciting the underlying the types proceeds in two steps, where in the first we extract all the information we can out of $\{x^{(i)},t^{(i)}\}_{i \in [n]}$. 
 
\begin{lemma}\label{tract_alpha}
  In every budgeting instance $\I$, for every type $\alpha \in \Delta^m \times\ \R$:
  \begin{enumerate}
      \item There exists a solution $(x^{(\alpha)},t^{(\alpha)})$ to the optimization problem 
      \begin{align*}
        \max_{(x,\ t)}\quad &v_\alpha(x,t)=\sum_j\alpha_{j}\theta_j(x_jB_t)-\alpha_{f}f(t)
      \text{~~~~ s.t. ~~~~}x \in \Delta^m,\ t \in [-\frac{B_0}{n},\infty);   \label{agents_opt_problem}
      \end{align*}
     \item $t^{(\alpha)} > \frac{B_0}{n}$  and $\frac{n\theta'_j(x^{(\alpha)}_jB_{t^{(\alpha)}})}{f'(t^{(\alpha)})}=\frac{\alpha_f}{\alpha_j}$ for all $j$ s.t. $x^{(\alpha)}_j>0$ 
  ;
  \item If $\ \lim_{z \to 0}\theta'_j(z)=\infty$ and  $x^{(\alpha)}_j=0$ then $\alpha_j=0$.
  \end{enumerate}
\end{lemma}
The proof is deferred to the \hyperlink{tract_alpha_proof}{appendix}. (2) and (3) provide us with $k \leq m$ linear conditions for every agent $i$. If $k<m$,  for each $j$ such that  $x^{(i)}_j \cdot B_{t^{(i)}}=0$ we present $i$ with the follow-up question: "How much are you willing to pay for adding an $\chi_j$ spending on public good $j$?" where $\chi_j$ is any fixed amount. For any number $\tau$ she responds, we extract $\frac{\alpha_{i,j}}{\alpha_{i,f}}$ from
\[\alpha_{i,j}\theta_j(\chi_j)=\alpha_{i,f}(f(t)-f(t+\tau))\] 
(If $\tau=0$ then $\alpha_{i,j}=0$). Note that here we do not ask the agent for her preferred \underline{consistent} budget decision, i.e. her response must not admit $\chi_j=n\tau$. Rather, we ask her for the price she is willing to pay for a specific good, that is public good $j$ at a level of $\chi_j$ spending. After presenting her with at most $m$ such follow-up questions, we have enough information to fully describe her type.\footnote{While presenting up to $m$ such follow-up questions to every agent may seem a bit tedious and  adding a significant cognitive load on them, it is plausible that not much will actually be needed, if any. We address that in section \ref{sec_strongIC}.}  
\begin{corollary}[Preferences Elicitation]\label{cor:pref_elict}
 The type profile $\vec{\alpha}$ can be fully extracted via the following steps.  
 \begin{enumerate}
     \item We ask all agents for their preferred budget decisions $(x^{(i)},t^{(i)}) \in \Delta^m \times \tr\  \forall i \in [n]$. 
     \item For every $i,j$ such that $x^{(i)}_j \cdot B_{t^{(i)}}>0$,
     \[\frac{\alpha_{i,f}}{\alpha_{i,j}}=\frac{\theta'_j(x^{(i)}_{j}(B_0+nt^{(i)}))}{f'(t^{(i)})}\]
     and, if $\ \lim_{z \to 0}\theta'_j(z)=\infty$ and  $x^{(i)}_j=0$ then $\alpha_{i,j}=0$.
     \item If $x^{(i)}_j\cdot B_{t^{(i)}}=0$ for some $i$ and $j$ and $\lim_{z \to 0}\theta'_j(z)<\infty$ , we ask agent $i$ for $\tau_{i,j}$, the maximum increase in $t^{(i)}$ she willing to add for financing 
        a $\chi_j$ spending on public good $j$, where $\chi_j$ could be any positive number. For every such $i$ and $j$,
        \[\frac{\alpha_{i,j}}{\alpha_{i,f}}=\frac{(f(t)-f(t+\tau_{i,j}))}{\theta_j(\chi_j)}\] 
    \item Each agent's type is given (decisively) by the $m$ linear equations above and the additional $\sum_j\alpha_{i,j}=1$.
 \end{enumerate}
\end{corollary}

 \begin{example}[Running example, cont.]\label{ex:run_opt}
 Every agent submits the optimal budget decision w.r.t. her valuation function, i.e. some $(x^{(i)},t^{(i)}) \in \ds$ that maximizes~\eqref{eq:v_runex}. Using the above Corollary, we infer every agent's underlying type $\alpha_i$ from 
$\frac{10(x^{(i)}_{j}\cdot 3t^{(i)})^{-1}}{(2\sqrt{t^{(i)}})^{-1}}=\frac{\alpha_{i,f}}{\alpha_{i,j}}$
if $x^{(i)}_j>0$, and $\alpha_{i,j}=0$ otherwise. For example, the voting profile on the left can only be induced by the preference profile $\alpha^{RE}$ on the right  (RE for Running Example):

     \medskip
     \centering
\begin{tabular}{l|c|c|c}
(votes)& $t$ & $X^{(i)}_1\ (x^{(i)}_1)$ & $X^{(i)}_2\ (x^{(i)}_2)$   \\
\hline
voter 1 & 69.4 & 145.8 (0.7) & 62.5 (0.3) \\
voter 2 & 26.3 & 0 (0)& 78.9 (1)\\
voter 3 & 44.4 & 66.6 (0.5) & 66.6 (0.5) \\
\end{tabular}
\quad\quad\quad
   \begin{tabular}{l|c|c|c}
(types) & $\alpha_{i,1}$ & $\alpha_{i,2}$ & $\alpha_{i,f}$   \\
\hline
$i=1$ & 0.7 & 0.3 & 0.8 \\
$i=2$ & 0 & 1 & 1.3\\
$i=3$ & 0.5 & 0.5 & 1 \\
\end{tabular}\vspace{.2cm}

 \end{example}

\newpar{Mechanisms} We now want to define a class of mechanisms for our budgeting problem. The first step in every such mechanism must be eliciting the types based on Corollary \ref{cor:pref_elict}. However, for the sake of a more convenient exposition, we formally define a class of \emph{Direct Revelation} \cite{nisan2007algorithmic} mechanisms that take the explicit type profile as input, which by Corollary \ref{cor:pref_elict} will bring to no loss of generality.    
\begin{definition}[mechanisms]
 A \emph{\textbf{mechanism}} for budgeting instance $\I$  is a pair $M=(\phi,P)$ where:
 \begin{itemize}
     \item $\phi: (\Delta^m \times \R )^n \mapsto \ds$ is a social choice function that inputs the type profile $\vec{\alpha}=(\alpha_1,\dots,\alpha_n)$
     and outputs a budget decision $(x,t) \in \ds$; 
     \item $P : (\ds)^n \mapsto \R^n$ is a function of the type profile that assigns a payment $P_i$ for every agent $i \in [n]$.
 \end{itemize}
 \end{definition} 
Note again that in our model, the outcome itself includes a tax payment $t$ that every agent pays to fund the public budget, and the individual payments $(P_i,\dots,P_n)$ are charged on top of that. Hereinafter we abuse notation a little when writing  $u_i((x,t),P_i)$ or $u_i(M),\ u_i(\phi,P)$ for the overall utility of agent $i$ from budget decision $(x,t)$ and payment $P_i$ (that are determined by $M=(\phi,P)$). While we defined $u_i$ earlier as a function of an allocation $X \in \R^m_+$ and a general monetary transfer $\delta \omega$, the interpretation should be clear as every pair of budget decision $(x,t)$ and  payment $P_i$ implicitly define an allocation $X=xB_t$ and an overall monetary transfer of $\delta \omega = t+P_i$.
\par

\newpar{Incentive Compatible mechanisms} Incentive Compatibility requires that no agent could benefit from reporting some false preferences $\alpha'_i \neq \alpha_i$ to the mechanism.
\begin{definition}
A  mechanism $M$ is \emph{\textbf{dominant strategy incentive compatible (DSIC)}}  if for all $i \in [n],\ \alpha_i \in \td,\ \vec{\alpha}_{-i} \in (\td)^{n-1}$ and every  $\alpha' \in \td $ s.t. $\alpha' \neq \alpha_i$,
$$u_i(M(\alpha_i,\vec{\alpha}_{-i}))\geq u_i(M(\alpha',\vec{\alpha}_{-i}))$$
If that inequality is strict for all $i,\ \alpha_i,\  \vec{\alpha}_{-i}$ and $\alpha'$ we say that $M$ is \emph{\textbf{strictly dominant strategy incentive compatible (SDSIC)}}.
\end{definition}

\section{Utility-Sensitive VCG}\label{sec_USVCG}
In this section section we present our proposed mechanism (Def. \ref{mec_def}) and discuss its properties. We start with the payment function $P$.
\subsection{Payments}\label{sec:adjusted _payments}
Essentially, the payments we define are VCG payments adjusted to our non quasi-linear setup.\cite{clarke1971multipart} 
In general, the social choice function in a VCG mechanism $\phi^{\mathsmaller{VCG}}$ outputs the socially optimal outcome $\Omega^*=argmax_{\Omega}\sum_{i}v_i(\Omega)$ and  payments are  defined as follows.
\begin{definition}[VCG payments]
$$p^{\mathsmaller{VCG}}_i=-\sum_{k\neq i}v_k(\Omega^*)+h(\vec{\alpha}_{-i})\quad \forall i \in [n]$$
\end{definition}
where  $h(\vec{\alpha}_{-i})$ could be any function of $\vec{\alpha}_{-i}$, the partial preferences profile submitted by all agents excluding $i$. The VCG model assumes \emph{Quasi-linear} utility functions, meaning that 
\begin{align*}
~~~~~~~~~(*) \quad  u_i\big(VCG(\vec{\alpha})\big)
&=v_i(\Omega^*)-p^{\mathsmaller{VCG}}_i=\sum_{k}v_k(\Omega^*)-h(\vec{\alpha}_{-i})
\end{align*}
The above expression is the key property on which the DSIC of a VCG mechanism relies. Since $\Omega^*$ maximizes $\sum_{k}v_k(\cdot)$ and $h(\vec{\alpha}_{-i})$ is independent of $i$'s vote, an agent can never (strictly) increase her utility by manipulating the outcome $\Omega^*$. In our model, however, the quasi-linearity assumption is  violated. If we write $\Omega^*:=(x^*,t^*)$ then our model assumes
$$u_i(M(\vec{\alpha}))=\sum_j\alpha_{i,j}\theta_j\big(x^*_jB_{t^*}\big)-\alpha_{i,f}f\big(t^*+P_i\big),$$
for any mechanism $M$. Clearly, na\"ively setting $P_i=p^{\mathsmaller{VCG}}_i$ will not result in anything as useful as $(*)$. However, our utility model does entail some significant information we can exploit. While individuals in our model no more exhibit the simple, linear relation between utility gains (or losses) that stem from monetary transfers and those that come from the chosen outcome itself, their true relation \emph{is} in fact described in the utility function, thanks to the introducing of $f$ in it. Relying on that, we can adjust the payments appropriately so that the key property $(*)$ is maintained. We do that in the next definition and lemma.

\begin{definition}[Utility-Sensitive VCG Payments]\label{def_payments}
Let $\Omega^*=(x^*,t^*) \in \argmax_{(x,t)}\sum_{i \in [n]}v_i(x,t)$ be the socially optimal budget decision. For every agent $i$, we define the \textbf{utility-sensitive VCG payment}  as 
\begin{equation}\label{eq:gen_VCG}
P_i=-t^* +f^{-1}\Big(f(t^*)+\frac{1}{\alpha_{i,f}}p^{\mathsmaller{VCG}}_i\Big).    
\end{equation}
\end{definition}
\begin{lemma}\label{k_prop}
 Let $\Omega^*=(x^*,t^*)$ be the social optimum and $P_i$ the utility sensitive VCG payment given in \ref{def_payments} above. Then 
 $u_i(\Omega^*,P_i)=\sum_kv_k(\Omega^*)-h(\vec{\alpha}_{-i})$.
\end{lemma}
\begin{proof}
 \begin{align*} 
u_i(\Omega^*,P_i)&=\sum_{j=1}^m \alpha_{i,j}\theta_j(x^*_jB_{t^*})-\alpha_{i,f}f(t^*+P_i)
 =\sum_{j=1}^m \alpha_{i,j}\theta_j(x^*_jB_{t^*})-\alpha_{i,f}f(t^*)-p^{\mathsmaller{VCG}}_i\\
 &=v_i(x^*,t^*)-p^{\mathsmaller{VCG}}_i= \sum_kv_k(x^*,t^*)-h(\vec{\alpha}_{-i})\qedhere
\end{align*}
\end{proof}

\subsection{Definition of the mechanism}
\begin{definition}[\textbf{Utility-Sensitive VCG}]\label{mec_def}
The \emph{Utility-Sensitive VCG} (US-VCG) mechanism $\M$ is a tax-involved PB mechanism defined by:
$$\forall \vec{\alpha} \in (\Delta^m \times \R )^n,\quad \M(\vec{\alpha}):=\big(\Omega^*,P\big)$$
where $\Omega^*=(x^*,t^*) \in \argmax_{\Omega}\sum_{i}v_i(\Omega)$ is the social optimum and $P$ is the utility-sensitive VCG payment assignment given in  Def. ~\ref{def_payments}.
\end{definition}

\subsection{Incentive Compatibility}\label{sec_IC}

\begin{proposition}
    The US-VCG is DSIC.
\end{proposition}
We omit a formal proof of that result since given Lemma \ref{k_prop}, it would proceed exactly as the DSIC proof for general VCG \cite{nisan2007algorithmic}, and very similarly to our SDSIC proof we later present. Some preparations are needed before that.  

\newpar{Mean-dependency}\label{mean_dep} 
 We here point out some very useful characteristic of our model that plays a major role in both of our main results
 . Namely, that the outcome $\Omega^*$ depends solely on the average of types reported by all agents. Let $\bar{\alpha}:=\frac{1}{n}\sum_{i \in [n]}\alpha_i$ denote the types mean. Now, simply changing the order of summation in the social welfare
\begin{align*}
  \sum_iv_i(x,t)&=\sum_i\Big[\sum_j\alpha_{i,j}\theta_j(x_jB_t)-\alpha_{i,f}f(t)\Big]
  =\sum_j\Big[\sum_i\alpha_{i,j}\theta_j(x_jB_t)-\alpha_{i,f}f(t)\Big]\\
 &=\sum_j n\bar{\alpha}_j\theta_j(x_jB_t)- n\bar{\alpha}_ff(t)=n\cdot v_{\bar{\alpha}}(x,t)
\end{align*}
 shows that:
\begin{observation}\label{mean_obs}
 In any budget decision $(x,t)$, the social welfare is given by $\sum_iv_i(x,t)=n\cdot v_{\bar{\alpha}}(x,t)$. Consequently, $\Omega^*=(x^*,t^*)$ maximizes the social welfare if and only if it maximizes $v_{\bar{\alpha}}(x,t)$, the valuation function defined by the average type $\bar{\alpha}$. 
\end{observation}

Note that in addition, the above observation means that the social optimum $\Omega^*$ is computed in $O(1)$ time w.r.t. $n$. As for  payments, the typical choice for $h(\alpha_{-i})$ is the "Clarke Pivot Rule" $h(\alpha_{-i})=\sum_{k \neq i}v_k(g(\ba_{-i}))$ \cite{clarke1971multipart}, meaning, charging every agent with the social welfare of others in her absence. With that choice, computing every agent's payment is as hard as computing the outcome $\Omega^*$.  

\begin{example}[Running example, cont.]\label{ex:run_select}{}
For the utility functions from Eq.~\eqref{eq:v_runex} 
we have $\argmax_{(x,t)}v_\alpha=\Big((\alpha_1,\alpha_2),\big( \frac{2\cdot 10}{\alpha_f}\big)^2\Big)$.
The average type in the type profile $\alpha^{RE}$ from Example~\ref{ex:run_opt}   
is $\ba=\big(0.4,0.6,1.03\big)$, and thus the budget decision chosen by $\M$ is $\Omega^*=\big((0.4,0.6),377\big)$. 
\end{example}

Following Observation \ref{mean_obs}, we can use the definition below for a more convenient presentation. 
\begin{definition}\label{g_map}
For all $\alpha \in \ds$, define $g(\alpha):\Delta^m\times \R \to \Delta^m\times \R$ such that
$$g(\alpha)=(x^{(\alpha)},t^{(\alpha)}) \in \argmax_{(x,t)} v_\alpha(x,t)$$
\end{definition}
Meaning, $g$ maps every preferences vector $\alpha=(\alpha_1,\dots,\alpha_m,\alpha_f)$ to an optimal budget decision w.r.t. the corresponding valuation function $v_\alpha$. If that optimum is not unique, $g$ chooses one arbitrarily. In some places we use this notation somewhat abusively, ignoring that indecisiveness in the specific choice of $g(\alpha)$. In particular, by Observation~\ref{mean_obs} $\Omega^*=g(\ba)$. \\

We are now almost ready to state our main result. That is, that the US-VCG mechanism satisfies the stronger SDSIC if we further assume that optimal points are  characterized by MRS conditions.
\begin{definition}[MRS conditions]\label{def_MRS}
 For any $\alpha \in \td$, an optimum $g(\alpha)=(x^{(\alpha)},t^{(\alpha)}) \in \argmax_{(x,t)}v_\alpha(x,t)$ satisfies the \emph{MRS conditions} if 
 \[\frac{n\theta'_j(x^{(\alpha)}_j \cdot B_{t^{(\alpha)}})}{f'(t^{(\alpha)})}=\frac{\alpha_f}{\alpha_j}\quad \forall \alpha_j >0,\, \text{ and } \alpha_j=0 \implies x^{(\alpha)}_j=0. \]
\end{definition}
That characterization of optimal points is a standard convention in economic literature \cite{krugman2008microeconomics,stiglitz1977theory,jain2007eisenberg,bjorvatn2002tax}, and we elaborate on its justification after proving the Theorem. The proof will take advantage of the fact that the MRS conditions define a unique $\alpha \in \td$ for any given $g(\alpha) \in \ds$, as the $m$ linear equations in \ref{def_MRS} above, along with $\sum_j\alpha_j=1$, have a unique solution. We also add the assumption that $\ba_j >0\, \forall j$, in other words that $\forall j\, \exists\, i \text{ s.t. } \alpha_{i,j} >0$ which is fair to assume.
\begin{theorem}
    Assume that $\ba_j >0\, \forall j$ and that all social optima $g(\ba)$ satisfy the MRS conditions. Then the US-VCG mechanism is SDSIC.
\end{theorem}

\begin{proof}
Fix $i$, $\alpha_i$ and $\bar{\alpha}_{-i}$. Note that by Observation \ref{mean_obs} the US-VCG mechanism outputs $g(\ba)$, and, by Lemma \ref{k_prop},
$$u_i(\M(\vec{\alpha}))=\sum_k v_k(g(\bar{\alpha}))-h(\vec{\alpha}_{-i})=n\cdot v_{\bar{\alpha}}(g(\bar{\alpha}))-h(\vec{\alpha}_{-i})$$
Now, assume that $i$ falsely reports $\alpha'_i \neq \alpha_i$. Inevitably, that shifts the mean preferences  to some $\bar{\alpha}' \neq \bar{\alpha}$, and the social optimum that $\M$ outputs to $g(\bar{\alpha}')$. We want to show now that $g(\bar{\alpha}')$ is certainly \textbf{not} an optimum of $v_{\ba}$. If $g(\bar{\alpha}')$ admits the MRS equations w.r.t. $\ba'$ then it does not w.r.t. $\ba$, as these equations define a unique type that solves them. Thus, by assumption, $g(\bar{\alpha}')$ is not an optimum of $v_{\ba}$. Otherwise, since by Lemma \ref{tract_alpha} interior optimal points satisfy the MRS conditions, we have that $g(\bar{\alpha}')$ is on the boundary. However, points on the boundaries (i.e. with $x_j=0$ for some $j$) cannot satisfy the MRS conditions w.r.t. $\ba$ because $\ba_j >0\, \forall j$, and thus $g(\ba')$ is not an optimum of $v_{\ba}$ in that case too. Therefore,   
\begin{align*}
    u_i(\M(\alpha'_i,\vec{\alpha}_{-i}))&=n\cdot v_{\bar{\alpha}}(g(\bar{\alpha}'))-h(\vec{\alpha}_{-i})\\
    &<n\cdot v_{\bar{\alpha}}(g(\bar{\alpha}))-h(\vec{\alpha}_{-i})=u_i(\M(\alpha_i,\vec{\alpha}_{-i}))\qedhere
\end{align*}
\end{proof}

\newpar{MRS characterization} First note that by Lemma \ref{tract_alpha}, internal optimal points always satisfy the MRS conditions. Even without further justifications, the conjecture that social optima implement every public good in some level, even minor, seems reasonable in many contexts. Especially, if we consider a high-level allocation of the budget between city departments or sectors that are unlikely to be completely dismissed in any society (e.g. education systems, infrastructure, culture activities and such). Technically, however, our model does not imply that. In this short discussion we give some further (and reasonable) assumptions that will suffice.
Basically, $g(\alpha)$ satisfies the MRS conditions unless $n\alpha_j\theta'_j(0) < \alpha_ff'(t^{(\alpha)})$ for some $\alpha_j >0$, meaning that a type $\alpha$ agent would rather not to fund public good $j$ by a further tax increase. Since $f'$ has a maximum at zero, demanding that $n\alpha_{i,j}\theta'_j(0)>\alpha_{i,f}f'(0)\ \forall i,j$---meaning that all agents wish to spend \emph{some} amount on each  and every public good---is sufficient to rule that out, and implies the same inequality for $\ba$ too. In particular, assuming  $\lim_{X_j \to 0}\theta'_j(X_j)=\infty\ \forall j$ promises, along with Assumption \ref{theta_f_asm}, that $g(\alpha)$ satisfies the MRS conditions for all types $\alpha$. That assumption, that can also be put as $\lim_{X_j \to 0}\pdv{\Theta}{X_j}=0\, \forall j$ holds in some very natural examples - e.g. for logarithmic or power functions, and also in the canonical Cobb-Douglas and Leontief utility models. Moreover, we should note that the  $\theta_j$ utilities input monetary investment, and not the public good itself. Meaning, we should properly interpret them as the composition of an underlying production function $\Phi_j(X_j)$ and a direct utility function $\zeta_j(\Phi_j)$. Now if $\Phi_j$ is a production function, the conventional Inada conditions \cite{inada1963two} include, inter alia, that $\lim_{X_j \to 0}\Phi'_j(X_j)=\infty$ and thus any increasing direct utility $\zeta_j$ with $\zeta'_j(0)>0$ would imply $\lim_{X_j \to 0}\theta'_j(X_j)=\lim_{X_j \to 0}\Phi'(X_j)\zeta'(\Phi(X_j))=\infty$. Under that assumption, agents allocate any given budget among all public goods such that $\alpha_{i,j} >0$, and thus optimal points satisfy MRS. Whether assuming the Inada conditions explicitly or not, the characterization of optimal points by MRS conditions is a widely accepted convention in public goods Economic literature \cite{krugman2008microeconomics,stiglitz1977theory,jain2007eisenberg,bjorvatn2002tax}.
Summing the above more formally, for an SDSIC mechanism we need the following.
 \begin{assumption}\label{asm_inada}
     We assume that either one of the following holds:
     \begin{enumerate}
         \item $n\alpha_{i,j}\theta'_j(0)>\alpha_{i,f}f'(0)\ \forall i,j$.
         \item $\lim_{X_j \to 0}\theta'_j(X_j)=\infty\ \forall j$. 
     \end{enumerate}
 \end{assumption}
 \begin{corollary}[SDSIC]\label{cor_strongtruth}
     The US-VCG mechanism is SDSIC in any budgeting instance that respects Assumption  \ref{asm_inada}.     
 \end{corollary}
 Note, moreover, that if Assumption \ref{asm_inada} holds then elicitation is completed with no follow-up questions needed.
\subsection{Manipulations By Coalitions}

In general, VCG mechanisms are known to be highly prone to group manipulations \cite{bachrach2010honor,conitzer2006failures}. While individuals cannot benefit from reporting false preferences when the reports of all others are fixed, a group of agents can sometimes coordinate their misreports in such way that each of them (or some at least) benefits due to the untruthful reports of others. The US-VCG is no different in that. However, the SDSIC property ensures that any such coalition would not be sustainable in the sense that the colluding agents cannot trust each other to follow the agreed scheme. Thus, it may suggest that such coalitions are not likely to form in the first place. That softer robustness demand where we allow for manipulating coalitions as long as they are unsustainable in the above sense is captured in the \emph{Coalition Proof Nash Equilibrium} (CPNE) solution concept \cite{bernheim1987coalition}. While the original term is quite involved, 
 the application to our context is intuitive: in an SDSIC mechanism, no sustainable coalition could exist since the individual unique best response, under any circumstances, is for every agent to report her true preference. We thus formulate here a simpler term that the US-VCG satisfies, and that implies CPNE.\footnote{We refer the reader to \cite{bernheim1987coalition} for the original CPNE definition.} 
 First, we define a manipulation by a coalition as a coordinated misreport by all its members, that benefits them. 
 \begin{definition}[Coalition manipulation]
 A \textbf{manipulation} by a \textbf{coalition} $S \subset [n]$ is a partial type profile $\alpha'_S=\{\alpha'_i\}_{i \in S}$ such 
 that $\alpha'_i \neq \alpha_i\ \forall i \in S$ and 
 \[u_i(M(\alpha'_S,\alpha_{-S}) \geq u_i(M(\alpha_S,\alpha_{-S})\ \forall i \in S.\]
 and there exists $i \in S$ for which the inequality is strict.
 \end{definition}
 Now, we demand that if such a manipulation exists then $\alpha'_i$ is not a best response for at least one agent in the coalition. 
 \begin{definition}["One Step Coalition-Proof"]
 We say that a mechanism is \emph{One Step Coalition-Proof} (OSCP) if for any manipulation $\alpha'_S$ by a coalition $S$, there exists $i \in S$ s.t.
 $\alpha'_i \notin BR_i(\alpha'_{S \setminus \{i\}},\alpha_{[n]\setminus S})$.
     
 \end{definition}
 \begin{corollary}\label{cor_OSCP}
      The US-VCG mechanism is OSCP in any budgeting instance that admits Assumptions \ref{theta_asm}, \ref{f_asm} and \ref{asm_inada}.
    (and consequently implements the social optimum as a CPNE).
 \end{corollary}
The claim follows trivially from SDSIC, which means that $BR_i(\alpha_{-i})=\{\alpha_i\}$ for all $i$ and $\alpha_{-i}$.

\section{Vanishing Payments Under Per-Capita Utilities}\label{sec_pconv}
In this section we show that payments become negligible in large populations. While these payments are essential for aligning incentives, charging additional money from voters would be undesired in a PB context. For the technical proofs, we will have to further specify the utility model so that it captures some important feature of divisible PB, which has been (justifiably) overlooked in past PB literature as well as in this work up to this point. That is, that the utility achieved from a given spending $X_j$ on some public good $j$ must also depend on the number of people that enjoy it, $n$.\footnote{In applications other than PB this may not be plausible. However, in such applications charging payments from agents could be acceptable.} The reason we only have to address that now is that in this section we analyse the asymptotic behavior of payments w.r.t. $n$, and in our model the overall budget $B_t = B_0+nt$ depends on it directly. The following example illustrates the problem. 
\begin{example}\label{ex_counter_pc}
 Consider a budgeting instance where $m=1, B_0=0, \theta(nt)=(nt)^p$ and $f(t)=t^q$ for some $0<p<q<1$. A type $\alpha_f$ agent thus maximizes her utility 
$u_{\alpha_f}=(nt)^p-\alpha_ft^q$ at \[t=\big[\frac{\alpha_f q}{p}\big]^{\frac{1-q}{1-p}}\cdot n^{1-q} \xrightarrow[n \to \infty]{} \infty\]
\end{example}

That is a very unlikely result, whatever that sole public good may be. There is no reason to expect that larger populations would wish to pay infinitely larger taxes, nor is it the situation found in reality. The model allows that because every tax unit payed by an individual is presumably "matched" $n-1$ times by others, thus making the substitution rate grow proportionally with $n$. However, while larger societies probably do have larger available resources, they are also likely to have greater needs. Just as a country's economic state  is conventionally measured by its GDP index, on the large scale the quality of public goods should be associated with \underline{spending per capita} rather than with nominal spending.\footnote{See for example \url{https://data.oecd.org/gga/general-government-spending.htm}, and \cite{stiglitz1977theory}.}
Hence, we now narrow down the definition of $\Theta_i(X)$ to 
\begin{definition}[per capita valuations of public expenditures]
   \[\Theta_i(X,n)=\sum_{j=1}^m\alpha_{i,j}\theta_j(X_j/n)=\sum_{j=1}^m\alpha_{i,j}\theta_j(x_j(b_0+t)) \text{  where  }b_0:=\frac{B_0}{n}.\]
\end{definition}
Note that all of our previous results follow through since for any fixed $n$, $\theta_j(X_j/n)$ is a particular case of $\theta_j(X_j)$.
Realistically, the dependency on $n$ might not be necessarily that we assumed, and we take $\theta_j(X_j/n)$ as a benchmark and relatively simple case of a more general class of functions of the form $\theta_j(X_j,n)$. While for some public goods $\frac{X_j}{n}$ may capture the relation adequately---for example, the quality of an education system surely depends on its resources per child---for others it may serve more as a large scale approximation---e.g., if the city offers free cloud services that allocates space equally among users, the total number of users affects each of them only to the point where the provided space exceeds their needs. Still, the benefit for users must be somehow connected the to the available space per user (which is determined by spending).\\

Some comments on the notations before we proceed. First, in this section we define $h(\vec{\alpha}_{-i})$ in the VCG payments as the conventional Clarke pivot-rule \cite{clarke1971multipart} function that charges a voter with the (normalized) social welfare of all others in her absence
$h(\vec{\alpha}_{-i})=\sum_{k \neq i}v_k(g(\bar{\alpha}_{-i}))$,
 making the VCG payments  
$$p^{\mathsmaller{VCG}}_i=-\sum_{k\neq i}v_k(g(\bar{\alpha}))+\sum_{k \neq i}v_k(g(\bar{\alpha}_{-i}))=(n-1)\Big(v_{\bar{\alpha}_{-i}}(g(\bar{\alpha}_{-i}))-v_{\bar{\alpha}_{-i}}(g(\bar{\alpha}))\Big)$$
where the second equality is by Observation \ref{mean_obs}. Next, we give an alternative representation for valuation functions that would ease the technical analysis significantly.
\begin{definition}[Alternative representation of valuation functions]
Define the vector valued function $\V: \Delta^m\times \R \to \R^{m+1}$
$$\V(x,t)=(\theta_1(x_1(b_0+t)),\dots,\theta_m(x_m(b_0+t)),-f(t))$$
For every $\alpha \in \Delta^m\times \R$, we write the valuation function $v_\alpha$ as the dot product of $\alpha$ and $\V$:
$$v_\alpha(x,t)=\sum_j\alpha_j\theta_j(x_j(b_0+t))-\alpha_ff(t)=\alpha \cdot \V(x,t)$$
\end{definition}
In these notations, the VCG payments are written as
$$p^{\mathsmaller{VCG}}_i=(n-1)\bar{\alpha}_{-i}\cdot\Big(\V(g(\bar{\alpha}_{-i}))-\V(g(\bar{\alpha}))\Big)$$
Our main results in this section (Theorems \ref{price_conv_1} and \ref{price_conv_2}) are  basically the convergence of that expression to zero, at different rates. Essentially, the conditions for convergence are that as $\bar{\alpha}_{-i} \to \bar{\alpha}$ with $n$,  $g(\bar{\alpha}_{-i}) \to g(\bar{\alpha})$ as well, and fast enough. In other words, they rely on the guarantees we can provide for $g$'s smoothness around the solution $g(\ba)$. In our running example, for instance, easy to check that $g$ is as smooth as you can wish for.

\begin{example}[Running example, cont.]\label{ex:run_g}{}
Consider again Example~\ref{ex:run}, where the utility model is 
$$v_\alpha(x,t)=\sum_{j=1}^2\alpha_j\ln(x_j\cdot t)-\alpha_f\sqrt{t}$$
We have shown that for any $\alpha \in \td$, $g(\alpha)=((\alpha_1,\dots,\alpha_m),(\frac{2}{\alpha_f})^2)$.  Thus $g(\alpha)$ is  continuously differentiable for all $\alpha \in \td$ (and in particular at $\ba$).
\end{example}

Coming up next are two preliminary lemmas that establish the continuity of $g$ at the solution $g(\ba)$ when that is uniquely defined, which is sufficient for Theorem \ref{price_conv_1} that follows. The proof of  \ref{lemma_bound_solutions} is deferred to the appendix.  

\begin{lemma}\label{lemma_bound_solutions}
Let $S \subset \Delta^m \times \R$ such that $\inf\{\alpha_f:\alpha \in S\}>0$. Then $\sup\{|g(\alpha)| : \alpha \in S\}<\infty$.
\end{lemma}

\begin{lemma}\label{sol_cont}
 For any given $\alpha \in \Delta^m\times\R$, if $v_\alpha$ has a unique global maximum then
 $$\lim_{\beta \to \alpha}g(\beta)=g(\alpha).$$
\end{lemma}

Before presenting the proof, note that (a) this statement is not obvious because we did not assume that $g$ is continuous, and (b) it holds for any function $g$ that follows Definition \ref{g_map}, i.e. the specific arbitrary choice of $g(\beta)$ in case $v_\beta$ has multiple optima is irrelevant.
\begin{proof}
By assumption, $0 < \alpha\cdot(\V(g(\alpha))-\V(g(\beta)))$. On the other hand, putting $\epsilon:=(\alpha-\beta)$,
\begin{align*}
   0 < \alpha\cdot(\V(g(\alpha))-\V(g(\beta)))&=\beta\cdot(\V(g(\alpha))-\V(g(\beta)))+\epsilon\cdot(\V(g(\alpha))-\V(g(\beta))\\
   & \leq \epsilon\cdot(\V(g(\alpha))-\V(g(\beta)))\leq |\epsilon|\cdot|\V(g(\alpha))-\V(g(\beta))| .
\end{align*}
where the second inequality is by definition of $g(\beta)$. By Lemma \ref{lemma_bound_solutions} $g(\beta)$ is bounded, and therefore $|\epsilon|\cdot |\V(g(\alpha))-\V(g(\beta))| \xrightarrow[\epsilon \to 0]{}0$. Ergo, 
$$\lim_{\beta \to \alpha}\alpha\cdot (\V(g(\alpha))-\V(g(\beta)))=\lim_{\beta \to \alpha}v_\alpha(g(\alpha))-v_\alpha(g(\beta))=0$$
Since $v_\alpha$ is continuous, $g(\alpha)$ is unique and $g(\beta)$ is bounded, it must be that $g(\beta) \to g(\alpha)$. 
\end{proof}

\subsection{Bounding individual payments} 
We need one more definition before stating our main result in this section.
\begin{definition}
   A \emph{characteristic triplet} in a budgeting instance $\I$ is $\sigma =(b_0,\mu,\bar{\alpha})$ where 
   \begin{itemize}
     \item $b_0:=\frac{B_0}{n}\geq 0$ is the non tax funded  budget source per capita.   
     \item $\bar{\alpha}:=\frac{1}{n}\sum_{k}\alpha_k \in \Delta^m \times \R$ is the mean preferences vector of all agents.
     \item $1/\mu<\alpha_{i,f}<\mu\ \forall i \in [n]$. 
 \end{itemize}
\end{definition}
\begin{theorem}\label{price_conv_1}
 Let $\sigma =(b_0,\mu,\bar{\alpha})$ such that $v_{\bar{\alpha}}$ has a unique global maximum at $g(\bar{\alpha})$. Then for every $\epsilon >0$ there exists $n_\epsilon(\sigma)$ such that in every budgeting instance with characteristic triplet $\sigma$ and $n > n_\epsilon(\sigma)$,
$$|P_i|<\epsilon\quad \forall i \in [n].$$ 
\end{theorem}
As stated, Theorem \ref{price_conv_1} means that prices vanish if the population is sufficiently large while not taking into account the likely possibility that in reality, new members that join a community might change it's characteristic parameters $b_0,\mu$ and $\bar{\alpha}$. That is, we are saying that in any given community with known parameters $(b_0,\mu,\bar{\alpha})$, prices will be arbitrarily small if the population is large enough. Thus, as there is no reason to assume some correlation between these parameters and the population's size, the theorem essentially implies that prices are likely to be small, even negligible, in larger societies.
\begin{proof}
The $p^{\mathsmaller{VCG}}$ payments are defined as the loss an agent imposes on all others by participating, and are therefore always non-negative. Now,  
\begin{align}
0\leq p^{\mathsmaller{VCG}}_i&=(n-1)\bar{\alpha}_{-i}\Big(\V(g(\bar{\alpha}_{-i}))-\V(g(\bar{\alpha}))\Big)\notag\\
&=(n-1)\bar{\alpha}\Big(\V(g(\bar{\alpha}_{-i}))-\V(g(\bar{\alpha}))\Big)+(n-1)(\bar{\alpha}_{-i}-\bar{\alpha})\Big(\V(g(\bar{\alpha}_{-i}))-\V(g(\bar{\alpha}))\Big)\notag\\
& \leq (n-1)(\bar{\alpha}_{-i}-\bar{\alpha})\Big(\V(g(\bar{\alpha}_{-i}))-\V(g(\bar{\alpha}))\Big)\notag\\
&\leq \frac{n-1}{n}|\bar{\alpha}_{-i}-\alpha_i||\V(g(\bar{\alpha}_{-i}))-\V(g(\bar{\alpha}))| \xrightarrow[\ba_{-i} \to \ba ]{}0\label{eq:p_bound}
\end{align}
Where the first inequality is by definition of $g(\bar{\alpha})$,  and in the second we used $\bar{\alpha}=\frac{n-1}{n}\bar{\alpha}_{-i}+\frac{1}{n}\alpha_i$ and Cauchy–Schwarz. Since $|\bar{\alpha}_{-i}-\alpha_i|$ is bounded, and by Lemma \ref{sol_cont} $\V\circ g$ is continuous, we get the convergence at the end. Now, as $\ba_{-i} \to \ba$ with $n \to \infty$, 
\begin{align*}
  P_i=-t(\bar{\alpha}) +f^{-1}\Big(f(t(\bar{\alpha}))+\frac{1}{\alpha_{i,f}}p^{\mathsmaller{VCG}}_i\Big) \xrightarrow[n \to \infty]{}  -t(\bar{\alpha})+t(\bar{\alpha}) =0
\end{align*}
(note that $f^{-1}$ is continuous). Thus for any arbitrary small $\epsilon$ we can find $n_\epsilon(\sigma)$ that yields the result.
\end{proof}
Note that the mean-dependency is crucial for that result too. It yields that $\ba_{-i}$ approaches $\ba$ at a $1/n$ rate, precisely canceling the increase in the number of agents $n$.
\subsection{Non-Positive Payments}\label{sec_conv2}
The theorem above shows that individual payments vanish with $n$. Our next goal is to formulate them as non positive for all agents. Meaning, we want no agent to add any payment on top of the tax $t$, even negligible. Instead, they might be paid a "negative payment" that we can view as a bonus or a "tax discount" for their participation. We defer the complete discussion and formalities to the \hyperlink{app_sec:np_payments}{appendix}, and only outline here the main ideas.

First, we formulate the condition needed for a stronger convergence result.
\begin{definition}
  Define  $F:(\Delta^m\times\R)^2 \to \R^{m+1}$ as:
\begin{align*}
   F_j(\alpha,(x,t))&=\alpha_j\theta'_j(x_j(b_0+t))-\alpha_ff'(t)\quad  \forall 1\leq j \leq m \\ 
    F_{m+1}(\alpha,(x,t))&=\sum_jx_j-1
\end{align*} 
 We say that $g(\alpha)=(x^{(\alpha)},t^{(\alpha)})$ is a \emph{\textbf{"regular maximum"}} of $v_\alpha$ if $F(\alpha,(x^{(\alpha)},t^{(\alpha)}))=0$ and
$$\det\left[\pdv{F}{(x,t)} \big(\alpha,g(\alpha)\big)\right]\neq 0 $$
where $\pdv{F}{(x,t)}$ is the $(m+1)\times(m+1)$ Jacobi matrix of $F$ with respect to the variables $(x,t)$.
\end{definition}
Next, in Appendix~\ref{apx:np_payments} 
we show that a unique regular solution $g(\ba)$ implies that $g$ is differentiable in its surrounding. That will enable us, by linear approximation, to evaluate the convergence rate of the difference $\big(\V(g(\bar{\alpha}_{-i}))-\V(g(\bar{\alpha}))\big)$, which essentially determines the payment of agent $i$. Theorem \ref{price_conv_2} then shows that this rate is $O(1/n)$. 

\begin{theorem}\label{price_conv_2}
 Let $\sigma =(b_0,\mu,\bar{\alpha})$ such that $v_{\bar{\alpha}}$ has a regular unique global maximum $g(\bar{\alpha})$. Then there exist some $\mathcal{B} \in \R$ and $n(\sigma)$ such that 
 in every population with characteristic triplet $\sigma$ and size $n > n(\sigma)$,
$|P_i| \leq\frac{\mathcal{B}}{n}$  for all $i \in [n]$. 
\end{theorem}

Now we construct non-positive payments the following way. We add to the VCG payments an amount that we pay back to every agent and  equals (or is greater than) the maximum payment she could have been charged with, given the partial type profile of her peers $\alpha_{-i}$. Theorem \ref{price_conv_2} will not only provide that bound, but also implies that the total amount paid to agents will not diverge as $n \to \infty$. Corollary \ref{cor_npp}  states the final result. 
\begin{corollary}
    In any budgeting instance with characteristic triplet $\sigma$ of size $n>n(\sigma)$ such that $v_{\bar{\alpha}}$ has a regular unique global maximum, there exist a payment assignment $P$ that satisfies: 
 \begin{enumerate}\label{cor_npp}
     \item $P_i \leq 0\ \forall i \in [n]$
     \item $\sum_{i \in [n]}P_i \geq -\mathcal{\tilde{B}}$ for some $\mathcal{\tilde{B}} \in \R_+$.
 \end{enumerate}
\end{corollary}
 
\section{Biased Mechanisms}
In this section, we expand the US-VCG definition to a class of mechanisms that insert a bias towards an arbitrary desired outcome or a set of outcomes. Indeed, the designer may have some goal in mind that she may want to balance with welfare, for example a particular project she wants to promote, a legacy allocation from previous years, or an allocation she sees as fair. 


The general form of mechanisms in that class follows the familiar \emph{affine-maximizer} generalization of a VCG mechanism. (See \cite{nisan2007algorithmic}, p. 228). We start with choosing a \emph{bias function} $\ff(x,t): \Delta^m\times \R \mapsto \R$, that in one way or another favours---that is, assigns higher values to---outcomes we see as desirable. Note that $\ff(x,t)$ must be independent of the realization of preferences $\vec{\alpha}=(\alpha_1,\dots,\alpha_n)$. Then, we generalize the US-VCG definition as follows.

\begin{definition}[\textbf{Biased Utility-Sensitive VCG}]\label{def_busvcg}
For any bias function $\ff(x,t): \Delta^m\times \R \mapsto \R$, we define the \emph{Biased Utility-Sensitive VCG} (BUS-VCG)
mechanism:
$$\hM(\vec{\alpha}):= (\hg(\ba),\hP)$$
where
\begin{itemize}
    \item $\forall \alpha \in \td,\ \hg(\alpha):= (x,t) \in \argmax_{(x,t)} v_\alpha(x,t)+\ff(x,t)\ $ \footnote{For the time being, let us just assume that indeed such a maximum exists. Note that $\hg$ Obviously depends on our choice for $\ff$, but we do not refer to that explicitly as it should be clear in the context.}
    \item $\hP$ is the payment assignment 
\[\hP_i:=-\th+f^{-1}\Big(f(\th)+\frac{1}{\alpha_{i,f}}(p^{\mathsmaller{VCG}}_i+n\ff(\hg(\ba_{-i}))-n\ff(\hg(\ba))\Big)\]
\end{itemize}
\end{definition}

\subsection{Equitable/Egalitarian Allocations}\label{sec_equit_all}
We now give an example for a specific choice of allocations that promotes two well-known notions of distributive justice and that a social planner might favor. We focus here only on the choice of the allocation $x \in \Delta^m$ and not on the tax decision $t$, assuming the designer only cares for the division of the budget $B_t$ once it is determined. Conceptually, we desire a "fair" use of public resources, whereas t is a private resource. There are also technical reasons for refering to the allocation alone, that will be clarified shortly.

\begin{definition}\label{def_equitall}
For every $t \in (-\frac{B_0}{n},\infty)$, let
$\he^t := \argmin_{x \in \Delta^m} \max_{j,k \in [m]} |\theta_j(x_jB_t)-\theta_k(x_kB_t)|$
\end{definition}
Thanks to the fact that $\Theta_i(xB_t)$ is a convex combination of $\theta_j(x_jB_t),\ j \in [m]$ for every agent $i$, $\he^t$ minimizes the possible difference in satisfaction between any two agents, up to a point of a complete \emph{equitable} allocation where $\theta_j(\he^t_jB_t)=\theta_k(\he^t_kB_t)\ \forall j,k$ and consequently $\Theta_i(\he^tB_t)=\Theta_h(\he^tB_t)$ for every two agents $i,h$, whenever that is feasible under $B_t$. (Such cases are not too scarce, for example if $\theta_j(0)$ is the same value for all $j \in [m]$ then an equitable allocation exists for all $t$.) We state that formally in the following proposition, given without a formal proof.  
\begin{proposition}\label{prop_equit}
For all $t \in (\frac{B_0}{n},\infty),\ \he^t= \argmin_{x \Delta^m} \max_{\alpha,\beta \in \Delta^m}\Big|\Theta_{\alpha}(xB_t)-\Theta_{\beta}(xB_t)\Big|$
\end{proposition}

Note that $\he^t$ optimizes with respect to all hypothetical types, as our bias function $\ff$ cannot depend on realized votes. Thus, diverting the outcome towards $\he^t$ brings an improvement in terms of the \underline{"worst case"} result only. For that reason, searching for an equitable or egalitarian budget decisions, i.e. w.r.t. the tax decision as well, also does not make too much sense on the technical level, as in every budget decision $(x,t)$, the worst-case maximum gap in $-\alpha_{i,f}f(t)$ is between two hypothetical agents that have the two extreme values of $\alpha_{i,f}$. Next, our choice of $\he^t$ implements (in the worst case) another celebrated principle in Social Choice literature, namely the Egalitarian Rule.

\begin{proposition}\label{prop:eg_all}
For all $t \in (-\frac{B_0}{n},\infty),\ \he^t= \argmax_{x \in \Delta^m} \min_{\alpha \in \Delta^m}\Theta_{\alpha}(xB_t)$.
\end{proposition}
\begin{proof}
 The minimum level of satisfaction in $(x,t)$ is attained by a unit-vector type $\ha$ such that 
 $\Theta_{\ha}(xB_t)=\theta_j(x_jB_t)=\min_{j'}\theta_{j'}(x_{j'}B_t)$. Now, it is clear that 
 \begin{align*}
    \he^t :&= \argmin_{x \in \Delta^m} \max_{j,k \in [m]} |\theta_j(x_jB_t)-\theta_k(x_kB_t)| 
           = \argmax_{x \in \Delta^m} \min_{j \in [m]} \theta_j(x_jB_t)
 \end{align*}
 \sloppy as one cannot increase both the minimum and maximum of \( \theta_{j'}(x_{j'}B_t),\ j' \in [m]\) while keeping the budget $B_t$ fixed, and thereby cannot increase the minimum without reducing the difference. 
\end{proof}

\subsection{Properties inherited by $\hM$}
We would obviously like to preserve useful properties of the US-VCG mechanism when generalizing to BUS-VCG. Some of them carry over quite easily. First, note that mean-dependency (Observation \ref{mean_obs}) holds for $\hg$ too. (We in fact rely on that in the definition of $\hM$). It is not difficult to see now that DSIC extends under any choice for $\ff$, as the definition of $\hP$ imitates the situation where $v_i(x,t) \to v_i(x,t)+\ff(x,t)$  and $h(\alpha_{-i}) \to h(\alpha_{-i})+n\ff(\hg(\ba_{-i}))$ for every agent $i$.\footnote{See \cite{nisan2007algorithmic} for a rigorous proof.} One also easily checks that payments vanish if $\hg(\ba)$ satisfies the same demands we needed earlier for $g(\ba)$, that is, continuity for Theorem \ref{price_conv_1} and differentiability for \ref{price_conv_2} (Albeit, for $\hg$ we will have to explicitly demand that while for $g$ we had other terms that implied its smoothness).
\begin{corollary}\label{cor_biasDSIC}
The BUS-VCG mechanism is DSIC. Moreover, Theorem \ref{price_conv_1} (\ref{price_conv_2}) hold if $\hg$ is continuous (differentiable) near the solution $\hg(\ba)$.  
\end{corollary}
SDSIC, however, will not carry over that easily. The condition we need for that is $\alpha \neq \alpha' \implies \hg(\alpha) \neq \hg(\alpha')$, which is not satisfied by any arbitrary choice of $\ff(x,t)$. For example,

\[\ff(x,t)=
\begin{cases}
K \ &(x,t)\in W\\
0 \  &(x,t)\notin W
\end{cases}\]
where $K \in R$ is some large number and $W \subset \ds$ is the set of favoured outcomes, does not satisfy that. In 
Appendix \ref{appsec_biasmech} we describe a class of bias functions that do preserve SDSIC.

\section{Concluding Remarks}\label{sec_conclude}
We presented a collective decision budgeting mechanism, the US-VCG mechanism, that concerns both the allocation and total volume of expenses. It is  essentially a VCG mechanism adjusted to our setting, in which we had to tackle a few issues. Mainly, we had to reformulate the payments to suit our preference model of non quasi-linear utilities. The US-VCG mechanism is welfare-maximizing and DSIC in the most general setup, and we specified the conditions in which it further satisfies  strict DSIC and consequently also resistance against coalition manipulations. In Section \ref{sec_pconv}, we showed that the modified VCG payments the mechanism charges become negligible in large populations, which is especially relevant in the Participatory Budgeting application that stood in the center of our focus. Finally, we showed a generalization of the US-VCG mechanism that inserts a bias towards any set of outcomes of one's choice.   

\newpar{Future Directions} In the introduction, we discussed the theoretic advantages of an additive concave utility model over other examples from the literature. The obvious downside is, when considering a mechanism that aggregates preferences, is the difficulty in assessing the concrete functions we should assume. While we can nevertheless argue that any such functions are probably a better approximation for the true underlying preferences than previous suggestions, future experimental research attempting to evaluate, similarly to those performed in relation to the disutility monetary function $f$ \cite{booij2010parametric}, could make a valuable contribution to the field.

\bibliographystyle{ACM-Reference-Format.bst}
\bibliography{Ref_EC23.bib}

\clearpage

\appendix

\section{Missing Proofs}\label{ap_sec:miss_prf}
\begin{rlemma}{tract_alpha}
   \hypertarget{tract_alpha_proof}{In} every budgeting instance $\I$, for every type $\alpha \in \td$:
  \begin{enumerate}
      \item There exists a solution $(x^{(\alpha)},t^{(\alpha)})$ to the optimization problem 
      \begin{align*}
        \max_{(x,\ t)}\quad &v_\alpha(x,t)=\sum_j\alpha_{j}\theta_j(x_jB_t)-\alpha_{f}f(t)\\
      s.t. \quad &x \in \Delta^m,\ t \in [-\frac{B_0}{n},\infty) 
      \end{align*}
     \item $t^{(\alpha)} > \frac{B_0}{n}$ and for all $j$ s.t. $x^{(\alpha)}_j>0$,  
  $$ \frac{n\theta'_j(x^{(\alpha)}_jB_{t^{(\alpha)}})}{f'(t^{(\alpha)})}=\frac{\alpha_f}{\alpha_j}$$  
   \item If $\lim_{z \to 0}\theta'_j(z)=\infty$ and  $x^{(\alpha)}_j=0$ then $\alpha_{i,j}=0$.
  \end{enumerate}
\end{rlemma}
\begin{proof}
For any preferences vector $\alpha$, consider first the (convex) optimization w.r.t. some fixed $t \in \R$, that is
 $$\max_{x \in \Delta^m}\sum_j\alpha_{j}\theta_j(x_jB_t)-\alpha_{f}f(t)$$
and let $x^{(\alpha)}(t)$ be the solution (that surely exists because $\Delta^m$ is a compact set) for any given $t \in \tr$. Then $x^{(\alpha)}(t)$ must satisfy the first order conditions: 
 \begin{align*}\label{x_foc}
    \alpha_{j}\theta'_j(x^{(\alpha)}_j(t)B_t)&= \alpha_{k}\theta'_k(x^{(\alpha)}_k(t)B_t)      
 \end{align*}
for every $x^{(\alpha)}_j(t),x^{(\alpha)}_k(t) >0$ and if $x_j>x_k=0$ then \[\alpha_{j}\theta'_j(x^{(\alpha)}_j(t)B_t) \geq \alpha_{k}\theta'_k(0)\]
because otherwise we can increase utility by some infinitesimal change of $x^{(\alpha)}(t)$. Note that (3) follows immediately from the above inequality if $\lim_{z \to 0}\theta'_k(z)=\infty$. We assume w.l.o.g. that $x^{(\alpha)}(t): \R \to \Delta^m$ is differential, because if not there exists a differentiable mapping arbitrarily close to $x^{(\alpha)}(t)$. 
 Now, consider
 \begin{align*}
     \dv{}{t}v_\alpha(x^{(\alpha)}(t),t)&=\dv{}{t}\left[\sum_j\alpha_{j}\theta_j(x^{(\alpha)}_j(t)B_t)-\alpha_ff(t)\right]\\
     &=\sum_j\alpha_j\theta'_j(x^{(\alpha)}_j(t)B_t)\Big[\dv{x^{(\alpha)}_j(t)}{t}B_t+nx^{(\alpha)}_j(t)\Big]-\alpha_ff(t)
 \end{align*}
 Now, there exists some $l \in [m]$ s.t. $x^{(\alpha)}_l(t) \geq 1/m$, and  by the first order conditions we derived above, 
 \begin{align*}
     \dv{}{t}v_\alpha(x^{(\alpha)}(t),t)&=\alpha_l\theta'_l(x^{(\alpha)}_l(t)B_t)\sum_j\Big[\dv{x^{(\alpha)}_j(t)}{t}B_t+nx^{(\alpha)}_j(t)\Big]-\alpha_ff(t)\\
     &\leq n\theta'_l(t/m)-\alpha_ff(t)
 \end{align*}
 where we used $\sum_j\dv{x^{(\alpha)}_j(t)}{t}=0$, $\sum_jx^{(\alpha)}_j(t)=1$ and $\theta''_l<0$.
 Now by Assumption \ref{theta_f_asm}, there exists some $t_0 \in \R$ s.t. $\forall j \in [m]$ 
 \[\theta'_j(B_t/m)\leq \theta'_j(t/m) < \frac{\alpha_f}{n}f'(t)\quad \forall t >t_0 \]
 and therefore $\dv{}{t}v_\alpha(x^{(\alpha)}(t),t)<0$ for $t >t_0$. Thus $v_\alpha(x^{(\alpha)}(t),t)$ has a global maximum at some $t^{(\alpha)} \leq t_0$. By Assumption \ref{theta_f_asm}, there exists some $j$ such that $\lim_{t \to 0}n\alpha_{j}\theta'(nt) > \alpha_{f}\lim_{t \to -\frac{B_0}{n}}f'(t)$. $n\theta'_j(0)$, and thus $t=-\frac{B_0}{n}$ must be suboptimal. In other words $t^{(\alpha)}>-\frac{B_0}{n}$.
 We abbreviate the notation now to $x^{(\alpha)}:= x^{(\alpha)}(t^{(\alpha)})$. We proceed while first assuming that $f$ is differentiable at $t^{(\alpha)}$. 
 \[\alpha_j\theta'_j(x^{(\alpha)}_j\cdot B_{t^{(\alpha)}})=\alpha_ff'(t^{(\alpha)})\quad \forall x^{(\alpha)}_j>0\]
because otherwise $v_\alpha(x^{(\alpha)},t^{(\alpha)}+\epsilon)>v_\alpha(x^{(\alpha)},t^{(\alpha)})$ for some infinitesimal (smaller or greater than zero) $\epsilon$.
To complete the proof, we now show that in case $\lim_{t \to 0^+_-}f'(t)=\infty$, $t^{(\alpha)} \neq 0$ and thus $f$ has a derivative at $t^{(\alpha)}$ in any case. Let $v_\alpha(x^{(\alpha)},t):=\sum_j\alpha_j\theta_j(x^{(\alpha)}_j\cdot B_t)-\alpha_ff(t)$, meaning, we fix $x^{(\alpha)}$ and define $v_\alpha(x^{(\alpha)},t)\ \forall t> \frac{B_0}{n}$. Then if $\lim_{t \to 0_-}f'(t)=\infty$, there exists some $\epsilon>0$ such that $\dv{}{t}v_\alpha(x^{(\alpha)},t)<0$ for all $t\in [-\epsilon,0)$ and because $v_\alpha(x^{(\alpha)},t)$ is continuous, $v_\alpha(x^{(\alpha)},0)<v_\alpha(x^{(\alpha)},-\epsilon)$. Thus $(x^{(\alpha)},0)$ cannot be a maximum. (Note that if $B_0=0$ then $t^{(\alpha)}=0$ is not possible by Assumption \ref{theta_f_asm}).
 
\end{proof}

\begin{rlemma}{lemma_bound_solutions}
Let $S \subset \Delta^m \times \R$ such that $\inf\{\alpha_f:\alpha \in S\}>0$. Then $\sup\{|g(\alpha)| : \alpha \in S\}<\infty$.
\end{rlemma}

\begin{proof}
 Note that $x^{(\alpha)} \in \Delta^m$ so we really only have to prove that $t^{(\alpha)}$ is bounded. By Lemma \ref{tract_alpha}, $$\frac{\theta'_j(x^{(\alpha)}_j(b_0+t^{(\alpha)}))}{f'(t^{(\alpha)})}=\frac{\alpha_f}{\alpha_j} \quad \forall j$$
For all $\alpha \in S$, there exists $j \in [m]$ such that $x^{(\alpha)}_j \geq 1/m$ and thus, since $\theta'_j$ is decreasing, 
$$\frac{\theta'_j(t^{(\alpha)}/m)}{f'(t^{(\alpha)})}\geq\frac{\theta'_j(x^{(\alpha)}_j \cdot t^{(\alpha)})}{f'(t^{(\alpha)})}\geq \frac{\theta'_j(x^{(\alpha)}_j(b_0+t^{(\alpha)}))}{f'(t^{(\alpha)})}=\frac{\alpha_f}{\alpha_j}$$
Therefore , by Assumption~\ref{theta_f_asm}, the LHS of that equality vanishes as $t^{(\alpha)} \to \infty$ and since the ratio on the RHS is bounded away from zero, $t^{(\alpha)}$ is bounded for all $\alpha \in S$.
\end{proof}


\section{Non-Positive Payments}\label{apx:np_payments}
We give here the full analysis that leads to Corollary \ref{cor_npp}. We start at formulating the condition needed for Theorem \ref{price_conv_2}.

\begin{definition}
  Define  $F:(\Delta^m\times\R)^2 \to \R^{m+1}$ as:
\begin{align*}
   F_j(\alpha,(x,t))&=\alpha_j\theta'_j(x_j(b_0+t))-\alpha_ff'(t)\quad  \forall 1\leq j \leq m \\ 
    F_{m+1}(\alpha,(x,t))&=\sum_jx_j-1
\end{align*} 
By Lemma \ref{tract_alpha}, $F(\alpha,(x^{(\alpha)},t^{(\alpha)}))=0$ for all $\alpha \in \Delta^m\times R$. We say that $g(\alpha)=(x^{(\alpha)},t^{(\alpha)})$ is a \emph{\textbf{"regular maximum"}} of $v_\alpha$ if, moreover,
$$\det\left[\pdv{F}{(x,t)} \big(\alpha,g(\alpha)\big)\right]\neq 0 $$
where $\pdv{F}{(x,t)}$ is the $(m+1)\times(m+1)$ Jacobi matrix of $F$ with respect to the variables $(x,t)$.
\end{definition}

The next Lemma is a direct application of the Implicit Function Theorem~\cite{wainwright2005fundamental}.
\begin{lemma}\label{lemma_IFT}
 Let $g(\alpha)$ be a regular maximum of $v_\alpha$. Then there exist an open  neighborhood of $\alpha$ $S \subset \Delta^m\times R$  and a unique mapping $s:S \to \ds$ such that :
\begin{itemize}
    \item $s(\alpha)=g(\alpha)$. 
    \item $s$ is continuously differentiable in $S$.
    \item $\forall \beta \in S,\ F(\beta,s(\beta))=0$  
\end{itemize}
\end{lemma}
The Lemma does not yet provide the smoothness we want for $g$, because it only promises that $\{s(\beta)\}_{\beta \in S}$ are critical points of $v_\beta$ that satisfy $F(\beta,s(\beta))=0$, not necessarily maxima. In other words $s$ and $g$ are not the same function by definition. However, an immediate corollary is that if $g(\alpha)$ is also unique they indeed must coincide.

\begin{corollary}\label{cor_sol_diff}
Assume that $g(\alpha)$ is a regular unique global maximum of $v_\alpha$. Then there exist a neighborhood of $\alpha$, $S \subset \Delta^m\times R$, such that $g$ is uniquely defined and continuously differentiable in $S$.
\end{corollary}
\begin{proof}
Denote by $\Tilde{S}$ the neighborhood of $\alpha$ from the statement of Lemma \ref{lemma_IFT} and let $g(\Tilde{S})$ be the image of $\Tilde{S}$ under $g$. By Lemma~\ref{sol_cont}, $g(\beta)$ approaches $g(\alpha)=s(\alpha)$ as $\beta \to \alpha$, therefore there exists a neighborhood $S \subseteq \Tilde{S}$ of $\alpha$ in which $g(\beta) \in g(\Tilde{S})\ \forall \beta \in S$. Now, every solution $g(\beta)$ must satisfy $F(\beta,g(\beta))=0$, but, by Lemma \ref{lemma_IFT} $s$ is the \emph{unique} function that maps $\beta \in \Tilde{S}$ to $s(\beta)$ such that $F(\beta,s(\beta))=0$, therefore $g\Big|_S=s$. Thus by Lemma \ref{lemma_IFT}, $g$ is uniquely defined and continuously differentiable in $S$. 
\end{proof}

We now have the necessary for the stronger version of payments convergence stated in Theorem~\ref{price_conv_2}. Basically, this further result is due to the differentiabilty of $g$ that allows for a linear approximation of the difference $\big(\V(g(\bar{\alpha}_{-i}))-\V(g(\bar{\alpha}))\big)$, which essentially determines the VCG payment of an agent $i$. Thus, we can not only say that this difference vanishes as we did in Theorem ~\ref{price_conv_1}, but also specify the convergence rate of $O(1/n)$. 

\begin{rtheorem}{price_conv_2}
 Let $\sigma =(b_0,\mu,\bar{\alpha})$ such that $v_{\bar{\alpha}}$ has a regular unique global maximum $g(\bar{\alpha})$. Then there exist some $\mathcal{B} \in \R$ and $n(\sigma)$ such that 
 in every population with characteristic triplet $\sigma$ and size $n > n(\sigma)$,
$$|P_i| \leq\frac{\mathcal{B}}{n}\quad \forall i \in [n]$$ 
\end{rtheorem}

\begin{proof}
 In the proof of Theorem~\ref{price_conv_1} we showed that:
 $$0 \leq p^{\mathsmaller{VCG}}_i\leq \frac{n-1}{n}|\bar{\alpha}_{-i}-\alpha_i||\V(g(\bar{\alpha}_{-i}))-\V(g(\bar{\alpha}))|$$
 By Corollary \ref{cor_sol_diff}, $g$ is differentiable in some neighborhood of $\bar{\alpha}$ $S$, and so $\V\circ g$ is too. Thus, if we take $n_0$  sufficiently large so that $\bar{\alpha}_{-i} \in S$, we have for all $n>n_0$ that
\begin{align*}
   p^{\mathsmaller{VCG}}_i&\leq \frac{n-1}{n}|\bar{\alpha}_{-i}-\alpha_i|\left|\mathcal{D}_{\V\circ g}(\bar{\alpha})\big(\bar{\alpha}_{-i}-\bar{\alpha})\big)+o|\bar{\alpha}_{-i}-\bar{\alpha}|\right| \\
   &\leq \frac{n-1}{n}|\bar{\alpha}_{-i}-\alpha_i| \Big(\norm{\mathcal{D}_{\V\circ g}(\bar{\alpha})}|\bar{\alpha}_{-i}-\bar{\alpha}|+o|\bar{\alpha}_{-i}-\bar{\alpha}|\Big)\\
   &=\frac{n-1}{n}|\bar{\alpha}_{-i}-\alpha_i| \Big(\norm{\mathcal{D}_{\V\circ g}(\bar{\alpha})}|\bar{\alpha}_{-i}-\alpha_i|\frac{1}{n}+o\big(|\bar{\alpha}_{-i}-\alpha_i|\frac{1}{n}\big)\Big)
\end{align*}
Where $\mathcal{D}_{\V\circ g}$ is the Jacobian matrix of $\V\circ g$ and $\norm{\mathcal{D}_{\V\circ g}}$ is its matrix-norm, and in the equality at the end we put $\bar{\alpha}=\frac{n-1}{n}\bar{\alpha}_{-i}+\frac{1}{n}\alpha_i$. Since the $\{\alpha_{i,f}\}$ are bounded there exists some $\gamma \in \R$ such that $|\bar{\alpha}_{-i}-\alpha_i|<\gamma$ for all $i$, thus
$$p^{\mathsmaller{VCG}}_i\leq\frac{n-1}{n^2}\gamma^2\Big(\norm{\mathcal{D}_{\V\circ g}(\bar{\alpha})}+o(1)\Big)$$
Now taking $n_1 \geq n_0$ such that the $o(1)$ term above is less than $1$ and ${\mathcal{B}_0=\gamma^2\Big(\norm{\mathcal{D}_{\V\circ g}(\bar{\alpha})}+1\Big)}$ gives 
\begin{align*}
  P_i=-t(\bar{\alpha}) +f^{-1}\Big(f(t(\bar{\alpha}))+\frac{1}{\alpha_{i,f}}p^{\mathsmaller{VCG}}_i\Big)<(f^{-1})'(t(\bar{\alpha}))\frac{\mu \mathcal{B}_0}{n}+o(\frac{\mu \mathcal{B}_0}{n})  
\end{align*}
(Note that by Lemma \ref{tract_alpha}, in case $f$ is not differentiable at zero then $t(\bar{\alpha}) \neq 0$ thus $(f^{-1})'$ alwyas exists.) And now taking $\mathcal{B}=\mu(f^{-1})'(t(\bar{\alpha}))\mathcal{C}_0+1$ with sufficiently large $n(\sigma) \geq n_1$ completes the proof.
\end{proof}

\begin{corollary}\label{pay_sum}
In any budgeting instance with characteristic triplet $\sigma$ and $n>n(\sigma)$, and such that $v_{\bar{\alpha}}$ has a regular unique global maximum,
$$\sum_{i \in [n]}P_i \leq \mathcal{C}$$
\end{corollary}

Besides telling us how fast payments converge, Theorem \ref{price_conv_2} and the corollary that follows allow for the definition of payments that not only vanish asymptotically, but are also non-positive for all agents. That is, instead of charging a "fee" for participating in the vote (and moreover, one that is not identical for everyone), agents will be paid by the mechanism, which you might see as a reward for their participation. The idea is quite simple---after charging an agent with her VCG payment,  we ``pay her back" no less than the maximum possible payment she could have been charged with given the preferences of her peers $\vec{\alpha}_{-i}$. This way no one is charged with a strictly positive payment. Thanks to Corollary \ref{pay_sum}, we know that the overall sum needed to implement that will not diverge with the number of agents.  

\begin{definition}[Non-positive utility-sensitive VCG payments]
  Define the \emph{non-positive VCG payment assignment} $\hat{p}^{\mathsmaller{VCG}}$ as
  $$\hat{p}^{\mathsmaller{VCG}}_i:=p^{\mathsmaller{VCG}}_i-\frac{\gamma^2}{n}\Big(\norm{\mathcal{D}_{\V\circ g}(\bar{\alpha}_{-i})}+1\Big)-\frac{r}{n}$$
  where $\gamma$ is the bound defined in the proof of Theorem \ref{price_conv_2} and $r \geq 0$ is any constant factor.
\end{definition}
  Note that this payment assignment does not violate SDSIC as it does not involve $\alpha_i$. \\
  $\frac{\gamma^2}{n}\Big(\norm{\mathcal{D}_{\V\circ g}(\bar{\alpha}_{-i})}+1\Big)$ is a bound we take for $p^{\mathsmaller{VCG}}_i$ (See clarification below) and the purpose of the $\frac{r}{n}$ factor is to allow for even higher payments to the agents, as much as the social planner wishes and can afford. We thus conclude that:

\begin{rcorollary}{cor_npp}
 In any budgeting instance with characteristic triplet $\sigma$ and $n>n(\sigma)$ such that $v_{\bar{\alpha}}$ has a regular unique global maximum, the payment assignment $\hat{P}_i:=-t(\bar{\alpha}) +f^{-1}\Big(f(t(\bar{\alpha}))+\frac{1}{\alpha_{i,f}}\hat{p}^{\mathsmaller{VCG}}_i\Big)$ satisfies: 
 \begin{enumerate}
     \item $\hat{P}_i \leq 0\ \forall i \in [n]$
     \item $\sum_{i \in [n]}\hat{P}_i \geq -\mathcal{\tilde{B}}$ for some $\mathcal{\tilde{B}} \in \R_+$.
 \end{enumerate}
\end{rcorollary}

Let us clear out these statements. First, $P_i \leq 0$ because $\hat{p}^{\mathsmaller{VCG}}_i \leq 0$ and $f^{-1}$ is increasing. Theoretically, we could have just put
    $\Big[\max_{\alpha_i \in \Delta^m \times \R}p_i^{\mathsmaller{VCG}}(\alpha_{-i},\alpha_i)\Big]$ instead of $\frac{\gamma^2}{n}\Big(\norm{\mathcal{D}_{\V\circ g}(\bar{\alpha}_{-i})}+1\Big)$ in the above definition, however we cannot guarantee that this maximum can be found efficiently. Note that $\frac{\gamma^2}{n}\Big(\norm{\mathcal{D}_{\V\circ g}(\bar{\alpha}_{-i})}+1\Big)$ is almost exactly the upper bound we put on $p^{\mathsmaller{VCG}}_i$ in the proof of Theorem \ref{price_conv_2}, the only difference is that we now use the derivative at $\bar{\alpha}_{-i}$ to evaluate the difference $V(g(\bar{\alpha}_{-i}))-V(g(\bar{\alpha}))$, instead of the derivative at $\bar{\alpha}$. Obviously, that is just as valid and we could have similarly reached that bound. The reason for this substitution is that the payments formula must be independent of $\alpha_i$ to maintain Incentive Compatibility. Now, While $\mathcal{D}_{\V\circ g}(\bar{\alpha}_{-i})$ is different for every agent $i$, we know that these are bounded globally because $\bar{\alpha}_{-i} \to \bar{\alpha}$ and the Jacobian matrix is a continuous function. Thus $\hat{p}_i^{\mathsmaller{VCG}}=O(p_i^{\mathsmaller{VCG}})=O(1/n)$.

\section{Bias Functions That Preserve SDSIC} \label{appsec_biasmech}

\subsection{Phantom-Agents}\label{sec:phant_bias}

We shall define our special bias function in the form of utility functions of fictitious agents that favour our targeted outcomes. In the simplest case where we have a sole budget decision in mind we can take $\ff(x,t)=\lambda\cdot v_{\tilde{\alpha}}(x,t)$ where $\tilde{\alpha}$ is chosen so that $g(\tilde{\alpha})$ is the desired outcome and $\lambda >0$ indicates the extent of bias we want. The SDSIC then extends trivially, as the our new biased mechanism is essentially the US-VCG mechanism for the $n$ real agents plus $\lambda n$ fictitious ones with mean preference $\tilde{\alpha}$. 
However, diverting the mechanism towards a set of outcomes while maintaining SDSIC will not be as simple (Note that if we add fictitious agents of multiple types, their impact on the outcome is ultimately determined only by the mean of fictitious preferences). 
\smallskip

The bias function we introduce below (Def. \ref{def_lambda_bfunc}) is composed as the sum of two functions, one that favours certain tax decisions and another that targets specific allocations $\hx^t$ for any given $t \in \tr$. That separation corresponds to the nature of our optimization problem, that can be solved in two steps accordingly (See details in the proof of Lemma \ref{tract_alpha}). Nonetheless, as we do not demand that $\hx^t$ necessarily exists for all $t$, i.e., it is possible that under some tax decisions no allocation is favourable,  it should not harm its generality. Meaning, any arbitrary set of outcomes $W \subset \ds$ can be targeted this way. We first define the follwing once $\{x^t\}_t$ are chosen.  

\begin{definition}
For any choice of $\{\hx^t\}_t$ and for every $t \in (-b_0,\infty)$ such that $\hx^t$ exists, we define the \emph{corresponding type} $\ha^t=(\ha^t_1,\dots,\ha^t_m)$ and  \emph{valuation function} $\hT^t$ such that
$$\argmax_{x \in \Delta^m}\hT^t(xB_t)=\argmax_{x \in \Delta^m}\sum_j\ha^t_j\theta_j(x_jB_t)=\hx^t$$
and if $\hx^t$ does not exist put $\ha^t=\vec{0}$.
\end{definition}

For all $t$, $\ha^t$ is derived uniquely  from $\hx^t$ using the first order conditions that we show in the proof of Lemma \ref{tract_alpha}:\footnote{We assume that $\hx^t$ is an internal point of $\Delta^m$ for all $t$.}

\[\frac{\theta'_j(\hx_jB_t)}{\theta'_k(\hx_kB_t)}=\frac{\ha^t_k}{\ha^t_j}\ \forall j \in [m]\]
(and note that these are linear equations as $\hx^t$ is given). Now, we define our bias function as follows.
 \begin{definition}[Phantoms Bias function]\label{def_lambda_bfunc}
For any choice of $\{\hx^t\}_t$ and for every $\lambda>0$, define   
$$\ff_{\lambda}(x,t):=\lambda \big(\hT^t(xB_t)-\hT^t(\hx^tB_t)\big)+ \psi(t)$$
where $\psi:(-\frac{B_0}{n},\infty) \mapsto \R$ is continuous and $\lim_{t \to \infty}\psi(t)=0$.
 \end{definition}
That is, for any given $t$ we add the (non-positive) utility loss of $\lambda$ fictitious agents that favour $\hx^t$ over $x$ (but have no preferences regarding the tax $t$), plus $\psi(t)$ that expresses the designer's preference on tax decisions. Note that by a similar argument to that used in the proof of Lemma \ref{tract_alpha}, that choice of $\ff$ assures that $\hg$ is well defined, i.e. that  $v_\alpha(x,t)+\ff_{\lambda}(x,t)$ has a global maximum for all $\alpha \in \ds$.
\smallskip

\begin{example}[Running example with bias to equitable allocation]\label{ex:run_eq}
We show an example for executing the BUS-VCG mechanism with $\hx^t$ taken as the equitable allocation $\he^t$ defined in \ref{def_equitall}. In case that $\theta_j$ are identical for all $j$ then $\he^t$ and $\ha^t$ are trivial,  $\he^t_j =\ha^t=1/m\  \forall j$. Thus,
\[C_\lambda(x,t)=\lambda\Big(10\sum_j\frac{1}{m}\ln(x_j\cdot 3t)-10\sum_j\frac{1}{m}\ln(\frac{1}{m}\cdot 3t)\Big)\]
To compute $\hg(\ha)=(\hx,\th)$, let us first find $\hx$ w.r.t. any fixed $t$. When $t$ is fixed, solving
\[\begin{cases}
    \max\ v_{\ba}(x,t)+3\ff_\lambda(x,t)\\
    \ x \in \Delta^m
\end{cases}\]  
is just as solving the original problem, only for a modified preferences mean $\beta_j=\frac{\ba_j+\lambda/m}{1+\lambda/m} \forall j \in [m]$, and with logarithmic functions we know that $\hx_j=\beta_j\ \forall j$, independently of $t$. Now we need to solve  
\[\begin{cases}
    \max\ v_{\ba}(\hx,t)+3\ff_\lambda(\hx,t)\\
    \ t \in \tr
\end{cases}\] 
but, note that $\ff(\hx,t)$ is a constant function of $t$, thus not affecting $\th$. To conclude, the introduction of $\ff_\lambda$ shifts the allocation from $x^*_j=\ba_j \forall j$ to $\hx_j=\frac{\ba_j+\lambda/m}{1+\lambda/m} \forall j $ while not affecting the tax decision  $t^*=\th=\big( \frac{2\cdot 10}{3\ba_f}\big)^2$.
\end{example}

As mentioned above, to maintain SDSIC we also have to see that $\hg(\alpha)$ defines $\alpha$ uniquely. That requires further assumption  on the smoothness of $\ff_\lambda$.  

\begin{lemma}\label{lem_biased_unisol}
 Assume that $\ff_\lambda(x,t)$ is differentiable in $t$. Then for any two distinct preferences vectors $\alpha\neq\beta \in \Delta^m\times \R$, ${\hg(\alpha) \neq \hg(\beta)}$.
\end{lemma}

\begin{proof}
  Fix any $\alpha \in \ds$, and let 
  $$\hg(\alpha)=(x^*,t^*) \in \argmax_{(x,t) \in \ds} v_\alpha(x,t)+\ff_{\lambda}(x,t).$$ 
  Then in particular, $x^*$ solves
  $\max_{x \in \Delta^m} v_\alpha(x,t^*)+\ff_\lambda(x,t^*)$.
  Since 
  \begin{align*}
    v_\alpha(x,t^*)&+\ff_{\lambda}(x,t^*)=\\
    &\sum_j\alpha_j\theta_j(x_jB_t^*)-\alpha_ff(t^*)
    +\lambda\sum_j\ha^{t^*}_j\theta_j(x_jB_t{^*})\\
    &-\lambda \sum_j\ha^{t^*}_j\theta_j(\hx^{t^*}_jB_{t^*})+\psi(t^*),
  \end{align*}
  $x^*$ also solves 
  $\max_{x \in \Delta^m} \sum_j(\alpha_j+\lambda\ha^{t^*}_j)\theta_j(x_jB_{t^*})$  because once we fix $t^*$ all the remaining terms are just constants.
  By the proof of Lemma \ref{tract_alpha}, that problem has a unique solution 
   for every $(\alpha_1,\dots,\alpha_m) \in \Delta^m$ and moreover, two distinct vectors cannot share the same solution. We thus conclude that $\hg(\alpha)=\hg(\beta) \implies \alpha_j=\beta_j\ \forall j \in [m]$.\\
   Now, let $\alpha,\beta \in \ds$ where $\alpha_j=\beta_j\ \forall j \in [m]$ and assume that $\hg(\alpha)=\hg(\beta)=(x^*,t^*)$. Then  
  \begin{align*}
      & t^* \in \argmax_t v_\alpha(x^*,t)+\ff_{\lambda}(x^*,t)\\
      \text{and }  & t^* \in \argmax_t v_\beta(x^*,t)+\ff_{\lambda}(x^*,t)
  \end{align*}
  Since $\alpha_j=\beta_j\ \forall j \in [m]$, we can write 
  \begin{align*}
     & v_\alpha(x^*,t)+\ff_{\lambda}(x^*,t)=\Gamma(t)-\alpha_ff(t) ;\\
     & v_\beta(x^*,t)+\ff_{\lambda}(x^*,t)=\Gamma(t)-\beta_ff(t) 
  \end{align*}
  for some $\Gamma: (-b_0,\infty) \mapsto \R$ that is differentiable by our  assumption on $\ff$ and the initial assumptions on $\theta_j,\ j \in [m]$.
 Thus,
  $$\Gamma'(t^*)-\alpha_ff(t^*)=\Gamma'(t^*)-\beta_ff(t^*)=0 \implies \alpha_f=\beta_f$$  
\end{proof}

\begin{corollary}
For all $\lambda \geq 0$ and for every bias function $\ff_\lambda(x,t)$  that is differentiable in $t$, the BUS-VCG mechanism is SDSIC in every budgeting instance that the US-VCG is. 
\end{corollary}

\section{Heterogeneous Tax}\label{appsec_heterotax}
In some situations we may prefer a heterogeneous tax distribution, especially in applications outside PB. For example, imposing higher contributions on wealthier countries in joint environmental investments.\footnote{Heterogeneous taxation is obviously appropriate as an income tax policy too, however it is less likely that private income levels are common knowledge in that scenario} We can implement a non-uniform distribution of tax revenues via assigning a parameter $\omega_i$ to each agent $i$, such that each pays $\omega_it$ and $\sum_i\omega_i =n$. The valuation function $v_i(x,t)$ is then reformulated to 
\[v_i(x,t) = \sum_{j=1}^m \alpha_{i,j}\theta_j\big(x_j \cdot (B_0+nt)\big)-\alpha_{i,f}f(\omega_it)\]

We show here how the US-VCG mechanism can still be implemented and with similar IC guarantees. (In particular, note that if $f$ is a power function then $f(\omega_it)=f(\omega_i)f(t)$ and thus we can absorb $\omega_i$ into $\alpha_{i,f}$, effectively changing nothing in the model and thus all results will follow through). As for preferences elicitation, the introduction of $\omega_i$ imposes no limitations on Corollary \ref{cor:pref_elict}. We only need to alter the individual equations from which we derive each agent's type accordingly.
In constructing the US-VCG mechanism, we shall define the payments assignment accordingly as
\[P_i=-t^* +f^{-1}\Big(f(\omega_it^*)+\frac{1}{\alpha_{i,f}}p^{\mathsmaller{VCG}}_i\Big)\]
and then Lemma \ref{k_prop} proceeds through, which is sufficient for DSIC. As for SDSIC, the social welfare is now expressed as
\[\sum_{j=1}^m \ba_j\theta_j\big(x_j \cdot B_t \big)-\sum_i\alpha_{i,f}f(\omega_it)\]
Meaning that mean dependency (Observation \ref{mean_obs} is lost. Nevertheless, Lemma \ref{lemma_bs} is built on the fact that a misreport of any agent necessarily shifts the outcome $g(\ba)$ to a sub-optimal point, thus bringing to a loss in her overall utility. Here we still have that, if the social optimum is characterized by MRS conditions, then
\[n\ba_j\theta'_j(x^*_jB_{t^*})=\sum_i\alpha_{i,f}f'(\omega_it^*)\cdot \omega_i\ \forall j\]
and thus if agent $i$ falsely report some $\alpha'_i$, then social optimum inevitably changes, and in particular to a sub-optimal point w.r.t the true social welfare. To see that, assume that w.l.o.g. $i$ reports $\alpha'_{i,f}>\alpha_{i,f}$. Then new (manipulated) type profile can still admit the above equations only if $\ba_j$ increases for all $j$, which is impossible.

\end{document}